\documentclass[letterpaper,journal]{IEEEtran}
\usepackage{amsmath,amsfonts}
\usepackage{array}
\usepackage{textcomp}
\usepackage{stfloats}
\usepackage{placeins}
\usepackage{url}
\usepackage{verbatim}
\usepackage{graphicx}
\def\BibTeX{{\rm B\kern-.05em{\sc i\kern-.025em b}\kern-.08em
    T\kern-.1667em\lower.7ex\hbox{E}\kern-.125emX}}
\usepackage{balance}
\usepackage{cite}

\usepackage{amssymb,amsthm,latexsym}
\usepackage{bm}
\usepackage{setspace}
\usepackage{tikz}
\usepackage[ruled]{algorithm2e}

\newcommand{\rmb}{\mathrm{b}}
\newcommand{\rmt}{\mathrm{t}}
\DeclareMathOperator{\tr}{tr}
\newtheorem{theorem}{Theorem}
\newtheorem{proposition}[theorem]{Proposition}

\graphicspath{{Figures/}}

\begin{document}

\title{Polarforming for MIMO Covert Communications}

\author{Xin Xie, \IEEEmembership{Member,~IEEE,}
	Jinpeng Xu, \IEEEmembership{Member,~IEEE,}
	Yihang Huang, \IEEEmembership{Member,~IEEE,} \\
    Li~Zhou,~\IEEEmembership{Senior Member,~IEEE,}
    Zhaolong~Ning,~\IEEEmembership{Senior Member,~IEEE}
	\thanks{
		(\textit{Corresponding author}: Jinpeng Xu.)

        Xin Xie, Yihang Huang, and Zhaolong Ning are with the School of Communications and Information Engineering, Chongqing University of Posts and Telecommunications, Chongqing 400065, China (e-mail: xiex@cqupt.edu.cn, huangyihang@cqupt.edu.cn, z.ning@ieee.org).

        Jinpeng Xu is with the Department of Electronic and Computer Engineering, The Hong Kong University of Science and Technology, Hong Kong SAR (e-mail: eejinpengxu@ust.hk).
        
		Li Zhou is with the College of Electronic Science and Technology, National University of Defense Technology, Changsha 410073, China (e-mail: zhouli2035@nudt.edu.cn).
	}
}

\maketitle

\begin{abstract}
Covert communication conceals wireless transmission activity, but conventional multi-antenna designs rely mainly on spatial beamforming and can be constrained by limited spatial degrees of freedom. This paper investigates polarization-reconfigurable antenna (PRA)-aided multiple-input multiple-output (MIMO) covert communication, where Alice and Bob perform transmit and receive polarforming, respectively, and Willie uses a vertically polarized antenna for radiometric detection. We jointly optimize the transmit covariance matrix and the transmit and receive phase shift vectors to maximize Bob's achievable rate under the covertness and transmit power constraints. Under Gaussian signaling, perfect polarized channel state information, and known noise power at Willie, we derive his optimal radiometric threshold and exact minimum detection error probability (DEP) for a finite observation length, and convert the DEP requirement into a deterministic constraint on his received signal power. We further prove that, for a fixed transmission design, positive signal leakage becomes detectable as the observation length tends to infinity. To solve the design problem for a finite observation length, we develop an alternating optimization (AO) algorithm with a generalized water-filling covariance update and closed-form phase updates. Numerical results validate the exact DEP analysis and the rapid convergence of the proposed algorithm. Compared with the sufficient divergence-based condition, the exact covertness constraint permits a $0.99$--$1.12$~dB higher received signal-to-noise ratio (SNR) at Willie, while joint transmit and receive polarforming provides a rate gain of up to $53.8\%$ over the considered benchmark schemes.
\end{abstract}

\begin{IEEEkeywords}
	Polarization-reconfigurable antennas, polarforming, covert communication, detection error probability.
\end{IEEEkeywords}

\section{Introduction}
Sixth-generation (6G) wireless networks are evolving into adaptive,
programmable ecosystems that integrate communications, sensing,
computation, and control and support applications such as embodied
intelligence, digital twins, and real-time industrial automation
\cite{6g_ai_2026,xu2026digital,Xu2026IntelligentFANETs}. Meanwhile, heterogeneous network architectures that integrate terrestrial,
aerial, satellite, and maritime infrastructures are further expanding the
scale and complexity of wireless connectivity
\cite{11207211,xu2026sagsin_outage_capacity}.
These developments make trustworthiness, encompassing security, privacy, and resilience, a fundamental system requirement \cite{10970466,11164329}. Among its privacy dimensions, unobservability is particularly important at the signal level. Although cryptographic protection \cite{10286081,10844934} and secrecy-oriented physical-layer security \cite{XJP_MIMOCovert,10949678,11534098} techniques can safeguard message content, they do not necessarily conceal the occurrence of a transmission, which may itself reveal sensitive communication activity.

Covert communication has been widely investigated to meet the transmission unobservability requirements of emerging wireless systems \cite{10440606,10195191,10938035,11365966}. Among the enabling techniques, multi-antenna beamforming is widely used to enhance reception at the legitimate receiver, Bob, while limiting the signal observed by the warden, Willie. However, conventional beamforming operates only in the spatial domain and may be constrained when the available spatial degrees of freedom (DoFs) are limited \cite{11235492,xu2026exploitinguavattitudecovert}.

Beyond the spatial domain, polarforming exploits the complementary polarization DoF by adapting the polarization state of each antenna. In practical propagation environments, reflection, diffraction, scattering, antenna misalignment, and imperfect cross-polar isolation can change the polarization state of an electromagnetic wave and cause polarization mismatch at the receiver \cite{Oestges_TWC_2008}. Conventional fixed-polarization antennas cannot adapt to these variations, whereas fully digital dual-polarized arrays generally require a separate radio-frequency (RF) chain for each polarization component. A phase-shifter-based polarization-reconfigurable antenna (PRA) offers a cost-effective implementation by combining two orthogonally polarized elements through a single RF chain and controlling their relative phase \cite{zhou2024polarforming,11347516}. Accordingly, polarforming can match the received polarization at an intended receiver or deliberately create polarization mismatch at an unintended receiver, without relying exclusively on spatial nulling.

These capabilities motivate further investigation of polarforming for multiple-input multiple-output (MIMO) covert communication. In particular, how spatial signaling should be coordinated with transmit and receive polarforming, and how the resulting design affects a warden's detection performance for a finite observation length, remain insufficiently understood. The main contributions of this paper are summarized as follows:
\begin{itemize}
    \item \emph{PRA-Aided MIMO Covert Communication Framework:} We establish a PRA-aided MIMO framework that jointly exploits transmit and receive polarforming for covert communication. The model incorporates antenna coupling, channel depolarization, cross-polar leakage, MIMO covariance signaling, and the vertically polarized projection observed by Willie. Based on this model, we jointly optimize the transmit covariance matrix and the transmit and receive phase shift vectors (PSVs) to maximize Bob's achievable rate.
    \item \emph{Exact Covertness Characterization:} Under Gaussian signaling and known noise power at Willie, we characterize covertness in two observation regimes. For a finite observation length, we derive Willie's optimal radiometric threshold and exact minimum detection error probability (DEP), whose monotonicity in his received signal-to-noise ratio (SNR) yields an equivalent received signal power constraint. As the observation length tends to infinity, we prove that any positive signal leakage becomes detectable and that a fixed transmission design satisfies a nontrivial covertness requirement only when the leakage is zero.
    \item \emph{Joint Transceiver Optimization:} For a prescribed finite observation length, we develop an alternating optimization (AO) algorithm that alternates among the transmit covariance matrix and the transmit and receive PSVs. The covariance update admits a generalized water-filling structure under the transmit power and received signal power constraints, while each phase is updated using closed-form candidates that preserve feasibility. The resulting objective sequence is nondecreasing and convergent.
    \item \emph{Numerical Validation and Design Insights:} Numerical results validate the exact DEP analysis and the rapid convergence of the proposed AO algorithm. They show that the exact covertness constraint is less conservative than the divergence-based alternative, while joint transmit and receive polarforming outperforms the considered benchmark schemes and better accommodates variations in the covertness requirement and polarized channel composition.
\end{itemize}

The remainder of this paper is organized as follows. Section~\ref{sec:related_work} reviews polarforming for communication and security and physical-layer covert communication design. Section~\ref{sec:system_model} presents the PRA-aided MIMO communication and detection models. Section~\ref{sec:covertness_analysis} derives the finite-observation-length and asymptotic covertness characterizations. Section~\ref{sec:optimization} develops the proposed joint optimization algorithm. Section~\ref{sec:numerical_results} presents numerical results, and Section~\ref{sec:conclusion} concludes the paper.

\emph{Notation:}
Bold uppercase and lowercase letters denote matrices and vectors, respectively. For a matrix $\bm A$, $\bm A^T$, $\bm A^H$, $\tr(\bm A)$, and $\det(\bm A)$ denote its transpose, conjugate transpose, trace, and determinant, respectively; $\bm A^{1/2}$ denotes the Hermitian square root when $\bm A$ is positive semidefinite. The operators $\operatorname{diag}(\cdot)$ and $\operatorname{blkdiag}(\cdot)$ construct diagonal and block-diagonal matrices, respectively, and $\odot$ denotes the Hadamard product. The notation $\bm A\succeq\bm0$ indicates that $\bm A$ is positive semidefinite, $\|\cdot\|$ denotes the Euclidean norm, $[x]^+\triangleq\max\{x,0\}$, and $j=\sqrt{-1}$. The set $\mathbb T$ denotes the phase circle with elements represented modulo $2\pi$. The notation $\mathcal{CN}(\bm0,\bm\Sigma)$ represents a circularly symmetric complex Gaussian distribution with zero mean and covariance matrix $\bm\Sigma$. The functions $\Gamma(\cdot)$, $\Gamma(\cdot,\cdot)$, and $\gamma(\cdot,\cdot)$ are the complete, upper incomplete, and lower incomplete Gamma functions, respectively.

\section{Related Work}\label{sec:related_work}
The present work combines polarforming with covert communication. Related studies in these two areas are reviewed below.

\subsection{Polarforming for Communication and Security}
Polarization diversity has long been exploited through dual-polarized channel modeling and antenna architectures \cite{Oestges_TWC_2008}. Unlike a conventional dual-polarized antenna that independently processes two polarization ports, phase-shifter-based polarforming uses a PRA to adapt the relative phase between orthogonally polarized elements while retaining one RF chain per antenna. The initial polarforming framework characterized the resulting MIMO channel capacity and optimized the phase shift vectors of the transmit and receive PRAs \cite{zhou2024polarforming}. Dynamic polarization matching and mismatching can further mitigate depolarization and adapt the effective channel without requiring separate RF chains for the two polarization components \cite{11347516}.

Recent work has extended polarforming beyond point-to-point rate maximization. Joint position and polarization optimization was studied for movable antenna systems \cite{11206360}, while a polarforming antenna capable of electrical polarization control and mechanical position/rotation adjustment was developed for sensing-assisted communications \cite{11165349}. Three-dimensional polarforming with three orthogonally polarized elements was further considered for uplink multiuser fairness in multipath channels \cite{11357497}. Other extensions include joint polarforming and precoding based on statistical channel state information (CSI) for massive MIMO low Earth orbit satellite links \cite{11558360}, and optimization that accounts for polarization in reconfigurable intelligent surface (RIS)-assisted flexible intelligent metasurfaces \cite{11547205}. These studies focus on communication or sensing performance and do not consider physical-layer security or covert communication.

Secure wireless communication via polarforming was studied in \cite{11235492}, where the transmit beamforming vector and phase shift vector are jointly optimized to maximize the secrecy rate in the presence of an eavesdropper. More recently, \cite{11603394} jointly optimized transmit beamforming, transmit polarforming phase shifts, and finite blocklength for a multiple-input single-output (MISO) link with a fixed-polarization receiver and warden to maximize covert throughput under a sufficient divergence-based covertness constraint. However, that study is limited to transmit polarforming and rank-one beamforming in the MISO setting, and its sufficient divergence-based condition does not characterize Willie's exact minimum DEP for a prescribed observation length. Receive polarforming and general MIMO covariance signaling are also not considered.

\subsection{Physical-Layer Covert Communication}
Covert communication strategies have been developed across the time, frequency, spatial, power, and modulation domains \cite{en_covert_2025}. Existing studies have examined how propagation environments and warden capabilities affect covert communication design. RIS-assisted covert communication has been analyzed jointly with secrecy in rate-splitting non-orthogonal multiple access (NOMA) systems \cite{10086651} and with ambient backscatter relaying in cooperative networks \cite{10923651}. Warden collusion has been studied by characterizing its impact on covert throughput \cite{10879382}, while location and interference uncertainty have been incorporated into air-to-ground covert communication design \cite{11503892}. Covert multichannel communication against interference-assisted proactive detection has also been studied through joint optimization of the number of subchannels and transmit power \cite{10521454}. Collectively, these studies show that covertness must account for both environmental uncertainty and increasingly capable wardens.

Another line of work enhances covert communication through array processing and reconfigurable antenna architectures. Multi-antenna covert communication has been investigated for finite-blocklength millimeter-wave links with spatially random wardens \cite{10185203}. Other array and precoding designs include random frequency-diverse arrays \cite{10244106}, finite-alphabet MIMO precoding \cite{11215610}, and joint communication and jamming beamforming for wardens at multiple possible positions \cite{11313487}. Movable antennas have also been used to jointly design sensing and covert communication \cite{11417854}, while extremely large fluid antenna systems have been considered for the joint optimization of covert communication, secrecy, and sensing \cite{11365966}. Together, these approaches extend antenna-domain covert design from conventional precoding to frequency, position, and aperture reconfiguration.

In aerial networks, the three-dimensional mobility of UAVs can induce
pronounced space--time--frequency non-stationarity in mmWave and massive-MIMO
aerial channels \cite{Xu2022UAVChannel}. Against such dynamic propagation
conditions, cooperative uncrewed aerial vehicle (UAV) jamming has been combined
with trajectory, power, and time allocation for covert communication
\cite{10195191}. Joint secret and covert communication has also been investigated
for multiuser multi-antenna uplink UAV systems through joint beamforming and UAV
trajectory optimization \cite{10938035}. Interference uncertainty in UAV-to-UAV
covert communication has been analyzed using stochastic geometry
\cite{11122485,He-CovertUAV-2026}, while dynamic scheduling and sensing-based informed jamming
have been developed for multi-UAV systems \cite{11214184}. Reliable covert
communication in NOMA-aided cognitive satellite--aerial--terrestrial integrated
networks has further been studied through aerial relay scheduling
\cite{11141469}. These works demonstrate that both dynamic aerial propagation
and mobility-aware resource allocation are important considerations for covert
communication in UAV networks.

Overall, covert antenna design remains dominated by spatial processing. Existing covert polarforming is limited to transmit polarization optimization in a MISO finite-blocklength model under a sufficient divergence-based constraint, leaving receive polarforming, MIMO covariance signaling, and exact finite-observation-length DEP analysis unaddressed.

\begin{figure*}[ht]
	\centering
	\includegraphics[width=1\linewidth]{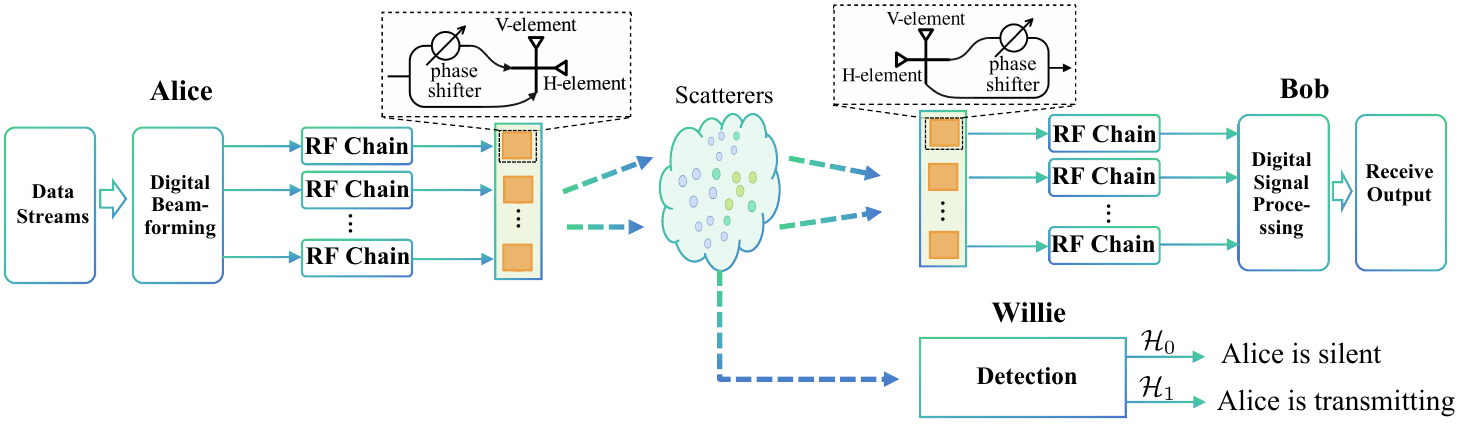}
	\caption{System model for PRA-aided MIMO covert communication.}
	\label{PRA_covert_System_Model}
\end{figure*}

\section{System Model}\label{sec:system_model}
As illustrated in Fig.~\ref{PRA_covert_System_Model}, we consider a PRA-aided MIMO covert communication system comprising a transmitter (Alice), a legitimate receiver (Bob), and a warden (Willie). Alice and Bob are equipped with $N_\rmt$ transmit PRAs and $N_\rmb$ receive PRAs, respectively. Willie employs a conventional single vertically polarized antenna and performs radiometric detection to determine whether Alice is transmitting. We first describe the phase-shifter-based PRA architecture and then establish the Alice-Bob communication model and Willie's binary detection model.

\subsection{PRA Model Based on Phase Shifters}
Alice and Bob employ phase-shifter-based PRAs. Each PRA is connected to a dedicated RF chain and comprises two orthogonally polarized elements, referred to as the vertical (V) and horizontal (H) elements. A phase shifter controls the relative phase between these elements modulo $2\pi$, thereby enabling the synthesis of linear, circular, and elliptical polarization states.
Let $\xi_{n_\rmt}\in\mathbb T$ and $\phi_{n_\rmb}\in\mathbb T$ denote the phase shifts of Alice's $n_\rmt$-th transmit PRA and Bob's $n_\rmb$-th receive PRA, respectively. The corresponding transmit and receive PSVs are ${\bm\xi}=[\xi_1,\xi_2,\ldots,\xi_{N_\rmt}]^T$ and ${\bm\phi}=[\phi_1,\phi_2,\ldots,\phi_{N_\rmb}]^T$, respectively. The polarforming vectors (PFVs) of the $n_\rmt$-th transmit PRA and the $n_\rmb$-th receive PRA are
\begin{align}
	{\bm f}(\xi_{n_\rmt}) &= \tfrac{1}{\sqrt{2}}
	\begin{bmatrix}
		1 ~ e^{j\xi_{n_\rmt}}
	\end{bmatrix}^T, \xi_{n_\rmt}\in\mathbb T,\ 1\le n_\rmt\le N_\rmt, \label{PFV_F_covert}\\
	{\bm g}(\phi_{n_\rmb}) &=
	\begin{bmatrix}
		1 ~e^{j\phi_{n_\rmb}}
	\end{bmatrix}^T, \phi_{n_\rmb}\in\mathbb T,\ 1\le n_\rmb\le N_\rmb, \label{PFV_G_covert}
\end{align}
where the factor $1/\sqrt{2}$ in Eq.~\eqref{PFV_F_covert} normalizes the transmit PFV.\footnote{No corresponding factor is required in Eq.~\eqref{PFV_G_covert}, because the receive PFV is not subject to the transmit power constraint.} Accordingly, the polarforming matrices of Alice and Bob are given by
\begin{align}
\!\!\!\!	{\bm F}({\bm\xi}) &\!=\! {\rm blkdiag}\left\{{\bm f}(\xi_1),{\bm f}(\xi_2),\ldots,{\bm f}(\xi_{N_\rmt})\right\}\in\mathbb{C}^{2N_\rmt\times N_\rmt},\\
\!\!\!\!	{\bm G}({\bm\phi}) &\!= \!{\rm blkdiag}\left\{{\bm g}(\phi_1),{\bm g}(\phi_2),\ldots,{\bm g}(\phi_{N_\rmb})\right\}\!\in\!\mathbb{C}^{2N_\rmb\times N_\rmb}.
\end{align}

\subsection{PRA-Aided MIMO Communication}
\subsubsection{Polarization Coupling Model}
In an ideal polarized link, the V and H components are perfectly isolated. In practice, imperfect antenna cross-polar isolation (XPI) and channel depolarization introduce residual coupling between them. Following \cite{zhou2024polarforming}, the coupling at the antennas is modeled by
\begin{align}
	\bm{\mathit{\Sigma}}_\mathrm{t} &= \frac{1}{\sqrt{1+\mu_{\rm t}}}
	\begin{bmatrix}
		1 & \sqrt{\mu_{\rm t}}\\
		\sqrt{\mu_{\rm t}} & 1
	\end{bmatrix},\\
	\bm{\mathit{\Sigma}}_\mathrm{b} &= \frac{1}{\sqrt{1+\mu_{\rm b}}}
	\begin{bmatrix}
		1 & \sqrt{\mu_{\rm b}}\\
		\sqrt{\mu_{\rm b}} & 1
	\end{bmatrix},
\end{align}
where $1/\mu_{\rm t}$ and $1/\mu_{\rm b}$ are the XPI values of Alice's and Bob's antennas, respectively, expressed on a linear scale. The inverse XPI parameters satisfy $0\leq\mu_{\rm t}\leq1$ and $0\leq\mu_{\rm b}\leq1$, which ensures that $\bm{\mathit{\Sigma}}_{\rm t}$ and $\bm{\mathit{\Sigma}}_{\rm b}$ are positive semidefinite and admit Hermitian matrix square roots.

Let $\tilde{\bm \Psi}_{\rm b}$ denote the polarization coupling matrix of the Alice-Bob channel, modeled as
\begin{align}
	\tilde{{\bm \Psi}}_{\rm b}= \frac{1}{\sqrt{1+\chi_{\rm b}}}
	\begin{bmatrix}
		1 & \sqrt{\chi_{\rm b}}\\
		\sqrt{\chi_{\rm b}} & 1
	\end{bmatrix},
\end{align}
where $\chi_{\rm b}$ is the inverse cross-polarization ratio (XPR) of the Alice-Bob link.

\subsubsection{Polarized MIMO Channel}
The polarized Alice-Bob channel is represented as
\begin{equation}
	\begin{aligned}
		\bm{P}_{\rm b}=
		\begin{bmatrix}
			\tilde{{\bm P}}_{11} & \cdots & \tilde{{\bm P}}_{1N_{\rmt}}\\
			\vdots & \ddots & \vdots\\
			\tilde{{\bm P}}_{N_{\rmb}1} & \cdots & \tilde{{\bm P}}_{N_{\rmb}N_{\rmt}}
		\end{bmatrix}
		\in\mathbb{C}^{2N_{\rmb}\times 2N_{\rmt}},
	\end{aligned}
	\label{Pb}
\end{equation}
where $\tilde{{\bm P}}_{n_\rmb n_\rmt}\in\mathbb{C}^{2\times 2}$ is the polarized subchannel between Alice's $n_\rmt$-th PRA and Bob's $n_\rmb$-th PRA. Under Rayleigh fading, this subchannel is modeled as
\begin{equation}
	\tilde{{\bm P}}_{n_\rmb n_\rmt}
	=
	\bm{\mathit{\Sigma}}_\mathrm{b}^{\frac{1}{2}}
	\left(
	\tilde{\bm\Psi}_{\rm b}\odot \tilde{\bm H}_{{\rm b},n_\rmb n_\rmt}
	\right)
	\bm{\mathit{\Sigma}}_\mathrm{t}^{\frac{1}{2}},
\end{equation}
where the entries of $\tilde{\bm H}_{{\rm b},n_\rmb n_\rmt}\in\mathbb{C}^{2\times 2}$ are independently and identically distributed as $\mathcal{CN}(0,\beta_{\rm b})$, and $\beta_{\rm b}$ accounts for the large-scale attenuation of the Alice-Bob link \cite{zhou2024polarforming}.

\subsubsection{Achievable Rate}
With polarforming at Alice and Bob, the effective channel from Alice's RF chains to Bob's RF chains is
\begin{equation}
	{\bm H}_{\rm b}({\bm\xi},{\bm\phi})\triangleq {\bm G}^H({\bm\phi}){\bm P}_{\rm b}{\bm F}({\bm\xi}).
\end{equation}

We consider quasi-static block fading over $L$ symbol intervals and assume that Alice employs Gaussian signaling. Specifically, the transmitted vectors are independently generated as $\bm x_l\sim\mathcal{CN}(\bm 0,\bm W)$, $l=1,\ldots,L$, where $\bm W\succeq\bm 0$ is the transmit covariance matrix. The average transmit power constraint is
\begin{equation}
	\mathbb E\{\|\bm x_l\|^2\}=\tr(\bm W)\le P.
\end{equation}
Since $\bm F^H(\bm\xi)\bm F(\bm\xi)=\bm I_{N_\rmt}$, this constraint is also the average radiated power constraint after transmit polarforming.
All channels remain constant within the $L$ symbol intervals.

Then, the received signal at Bob in the $l$-th symbol interval is
\begin{equation}
	\bm y_{{\rm b},l}={\bm H}_{\rm b}({\bm\xi},{\bm\phi})\bm x_l+\bm n_{{\rm b},l},\quad l=1,\ldots,L,
\end{equation}
where $\bm n_{{\rm b},l}\sim\mathcal{CN}(\bm 0,\sigma_{\rm b}^2\bm I_{N_\rmb})$ denotes the effective additive white Gaussian noise at Bob. Following \cite{zhou2024polarforming}, we assume that the noise introduced by the antenna elements and phase shifters is negligible relative to that generated by the active RF chain components. Accordingly, the dominant receiver noise is modeled after receive polarforming and has covariance $\sigma_{\rm b}^2\bm I_{N_\rmb}$. The achievable rate at Bob is
\begin{equation}\label{rate_bob}
	R_{\rm b}=\log_2\det\left(\bm I_{N_\rmb}+\frac{{\bm H}_{\rm b}({\bm\xi},{\bm\phi})\bm W{\bm H}_{\rm b}^H({\bm\xi},{\bm\phi})}{\sigma_{\rm b}^2}\right).
\end{equation}

\subsection{Covert Detection}

We assume that Alice has perfect instantaneous CSI for both the Alice-Bob and polarized Alice-Willie links. This idealized assumption yields an upper bound on the achievable covert communication performance.
For a conservative covertness assessment, Willie is assumed to know $\sigma_{\rm w}^2$ and the signal variance under $\mathcal H_1$, and therefore implements the optimal likelihood-ratio test for the stated model.

\subsubsection{Alice-Willie Channel}
Willie's fixed vertically polarized antenna cannot adapt its polarization state and therefore observes only the vertical projection of Alice's transmitted signal. Since Willie has only one vertically polarized receive port, no receive-side cross-polar coupling or XPI matrix arises at Willie. Nevertheless, imperfect transmit XPI at Alice and channel depolarization allow energy radiated by Alice's H elements to reach Willie's V receive port. Consequently, Willie's scalar observation contains both co-polar V-to-V propagation and cross-polar H-to-V leakage, where the latter is an effective propagation and transmit-side leakage term rather than crosstalk between receive polarization ports at Willie.
Let ${\bm h}_{\rm w}^H\in\mathbb{C}^{1\times 2N_\rmt}$ denote the equivalent polarized Alice-Willie channel observed at Willie's vertically polarized antenna, which can be written as
\begin{equation}
	\bm h_{\rm w}^H=
	\begin{bmatrix}
		u_1 & v_1 & \cdots & u_{N_\rmt} & v_{N_\rmt}
	\end{bmatrix},
\end{equation}
where $u_n$ and $v_n$ denote, respectively, the effective co-polar V-to-V coefficient and cross-polar H-to-V leakage coefficient associated with Alice's $n$-th transmit PRA after accounting for Alice's transmit-side XPI and channel depolarization. Since Willie has only one V receive port, no receive-side polarization coupling matrix is applied. After transmit polarforming, the effective Alice-Willie channel is
\begin{equation}
	{\bm q}_{\rm w}^H({\bm\xi})
	\triangleq
	{\bm h}_{\rm w}^H{\bm F}({\bm\xi})
	\in\mathbb{C}^{1\times N_\rmt}.
\end{equation}

\subsubsection{Binary Hypothesis}
Accordingly, Willie's received signal at the $l$-th symbol interval is
\begin{equation}
	y_{{\rm w},l}=
	\begin{cases}
		n_{{\rm w},l}, & \mathcal{H}_0,\\
		{\bm q}_{\rm w}^H({\bm\xi})\bm x_l+n_{{\rm w},l}, & \mathcal{H}_1,
	\end{cases}
	\quad l=1,\ldots,L,
\end{equation}
where $\mathcal{H}_0$ and $\mathcal{H}_1$ denote the hypotheses that Alice is silent and active, respectively. The noise terms $n_{{\rm w},l}\sim\mathcal{CN}(0,\sigma_{\rm w}^2)$ are independent across symbol intervals and independent of the transmitted signal vectors $\bm x_l$.
Under the Gaussian random signaling assumption, i.e., $\bm x_l\sim\mathcal{CN}(\bm 0,\bm W)$, the signal component observed at Willie satisfies
\begin{equation}
	{\bm q}_{\rm w}^H({\bm\xi})\bm x_l
	\sim
	\mathcal{CN}\left(
	0,
	{\bm q}_{\rm w}^H({\bm\xi})\bm W{\bm q}_{\rm w}({\bm\xi})
	\right).
\end{equation}
Define Willie's average received signal power under $\mathcal H_1$ as
\begin{equation}\label{willie_power}
	S_{\rm w}
	\triangleq
	{\bm q}_{\rm w}^H({\bm\xi})\bm W{\bm q}_{\rm w}({\bm\xi}).
\end{equation}
Because the transmitted signal and receiver noise are independent, Willie's observations in each symbol interval follow
\begin{align}
	y_{{\rm w},l}
	&\sim
	\mathcal{CN}(0,\sigma_{\rm w}^2),
	\quad \mathcal{H}_0,\\
	y_{{\rm w},l}
	&\sim
	\mathcal{CN}(0,\sigma_{\rm w}^2+S_{\rm w}),
	\quad \mathcal{H}_1.
\end{align}

\subsubsection{Radiometric Detection}
Because the two hypotheses differ only in variance, the sample energy is a sufficient statistic for the likelihood-ratio test. Willie's radiometric statistic is therefore
\begin{equation}\label{radiometer}
	T=\frac{1}{L}\sum_{l=1}^{L}|y_{{\rm w},l}|^2.
\end{equation}
Given a decision threshold $\eta$, Willie decides
\begin{equation}
	T \mathop{\gtrless}_{\mathcal{H}_0}^{\mathcal{H}_1} \eta.
\end{equation}

The conditional means of the test statistic are
\begin{align}
	\mathbb{E}[T|\mathcal{H}_0] &= \sigma_{\rm w}^2,\\
	\mathbb{E}[T|\mathcal{H}_1] &= S_{\rm w}+\sigma_{\rm w}^2.
\end{align}
For a threshold $\eta$, the false-alarm and missed-detection probabilities are, respectively,
\begin{align}
	P_{\rm FA}(\eta)&=\Pr(T>\eta|\mathcal{H}_0),\\
	P_{\rm MD}(\eta)&=\Pr(T<\eta|\mathcal{H}_1).
\end{align}
\subsubsection{Covertness Metric}
Following the covert communication convention, the minimum total DEP is defined as
\begin{equation}\label{DEP}
	\zeta_L^\star\triangleq \min_{\eta\ge0}\left\{P_{\rm FA}(\eta)+P_{\rm MD}(\eta)\right\}.
\end{equation}
A larger $\zeta_L^\star$ indicates stronger covertness. We impose the covertness requirement
\begin{equation}\label{covert_constraint}
	\zeta_L^\star\ge 1-\epsilon,
\end{equation}
where $\epsilon\in(0,1)$ is a prescribed covertness parameter. This metric weights false alarms and missed detections equally; equivalently, $\zeta_L^\star/2$ is the minimum Bayes error under equal priors.

We use Eq.~\eqref{DEP} as the covertness metric throughout this paper. The next section derives the exact finite-observation-length DEP and converts Eq.~\eqref{covert_constraint} into an equivalent deterministic constraint on Eq.~\eqref{willie_power}.

\section{Covertness Analysis and Problem Formulation}\label{sec:covertness_analysis}
In this section, we first characterize Willie's optimal detection performance for a finite observation length and convert the resulting DEP requirement into an equivalent received signal power constraint. We then examine the asymptotic regime as the observation length tends to infinity and identify the corresponding condition for covertness. Finally, using the finite-observation-length covertness characterization, we formulate the joint optimization of the transmit covariance matrix and the transmit and receive PSVs.

\subsection{Finite-Observation-Length Covertness Analysis}

We first characterize the exact DEP for a finite observation length $L$, without invoking an asymptotic approximation. Based on Eq.~\eqref{willie_power}, define Willie's received SNR as
\begin{equation}
	\rho \triangleq \frac{S_{\rm w}}{\sigma_{\rm w}^2}.
\end{equation}
Under $\mathcal H_0$ and $\mathcal H_1$, respectively, Willie's observation follows
\begin{align}
	y_{{\rm w},l} &\sim \mathcal{CN}(0,\sigma_{\rm w}^2), \quad \mathcal H_0,\\
	y_{{\rm w},l} &\sim \mathcal{CN}(0,\sigma_{\rm w}^2+S_{\rm w})
	=\mathcal{CN}(0,\sigma_{\rm w}^2(1+\rho)), \quad \mathcal H_1.
\end{align}
It follows from Eq.~\eqref{radiometer} that the normalized test statistic satisfies
\begin{align}
	\frac{LT}{\sigma_{\rm w}^2}
	&\sim \mathrm{Gamma}(L,1), \quad \mathcal H_0,\\
	\frac{LT}{\sigma_{\rm w}^2(1+\rho)}
	&\sim \mathrm{Gamma}(L,1), \quad \mathcal H_1,
\end{align}
where $\mathrm{Gamma}(L,1)$ denotes the Gamma distribution with shape parameter $L$ and scale parameter $1$.

For a given threshold $\eta$, Willie decides $\mathcal H_1$ when $T>\eta$ and $\mathcal H_0$ otherwise. The corresponding error probabilities are
\begin{align}
	P_{\rm FA}(\eta)
	&=
	\Pr(T>\eta|\mathcal H_0) \notag\\
	&=
	\frac{\Gamma\left(L,\frac{L\eta}{\sigma_{\rm w}^2}\right)}{\Gamma(L)},
	\label{PFA_exact}\\
	P_{\rm MD}(\eta;\rho)
	&=
	\Pr(T<\eta|\mathcal H_1) \notag\\
	&=
	\frac{\gamma\left(L,\frac{L\eta}{\sigma_{\rm w}^2(1+\rho)}\right)}{\Gamma(L)}.
	\label{PMD_exact}
\end{align}

To obtain the threshold minimizing $P_{\rm FA}+P_{\rm MD}$, define the observation vector $\bm y_{\rm w}\triangleq[y_{{\rm w},1},\ldots,y_{{\rm w},L}]^T$. Its likelihood functions under $\mathcal H_0$ and $\mathcal H_1$ are
\begin{align}
	p_0(\bm y_{\rm w})
	&=
	\frac{1}{(\pi\sigma_{\rm w}^2)^L}
	\exp\left(-\frac{\|\bm y_{\rm w}\|^2}{\sigma_{\rm w}^2}\right),\\
	p_1(\bm y_{\rm w})
	&=
	\frac{1}{(\pi\sigma_{\rm w}^2(1+\rho))^L}
	\exp\left(-\frac{\|\bm y_{\rm w}\|^2}{\sigma_{\rm w}^2(1+\rho)}\right).
\end{align}
Under equal priors, the likelihood-ratio test decides $\mathcal H_1$ when $p_1(\bm y_{\rm w})>p_0(\bm y_{\rm w})$, or equivalently when
\begin{equation}
	\|\bm y_{\rm w}\|^2
	>
	L\sigma_{\rm w}^2
	\frac{1+\rho}{\rho}
	\ln(1+\rho).
\end{equation}
Since $T=\|\bm y_{\rm w}\|^2/L$, the optimal radiometric threshold is
\begin{equation}
	\eta^\star(\rho)
	=
	\sigma_{\rm w}^2
	\frac{1+\rho}{\rho}
	\ln(1+\rho),
	\quad \rho>0.
	\label{eta_star_exact}
\end{equation}
For $\rho=0$, the two hypotheses induce identical distributions and are therefore statistically indistinguishable.

Substituting Eq.~\eqref{eta_star_exact} into Eqs.~\eqref{PFA_exact} and \eqref{PMD_exact} gives the exact minimum DEP as a function of $\rho$:
\begin{align}
	\zeta_L^\star(\rho)
	&=
	P_{\rm FA}(\eta^\star(\rho))
	+
	P_{\rm MD}(\eta^\star(\rho);\rho) \notag\\
	&=
	\frac{
		\Gamma\left(
		L,
		L\frac{1+\rho}{\rho}\ln(1+\rho)
		\right)
	}{
		\Gamma(L)
	}
	+
	\frac{
		\gamma\left(
		L,
		L\frac{\ln(1+\rho)}{\rho}
		\right)
	}{
		\Gamma(L)
	}.
	\label{DEP_exact}
\end{align}
At $\rho=0$, this function is defined by continuity as
\begin{equation}
	\zeta_L^\star(0)=1.
\end{equation}

Eq.~\eqref{DEP_exact} shows that $\bm W$ and $\bm\xi$ affect the DEP only through $S_{\rm w}$, or equivalently through $\rho=S_{\rm w}/\sigma_{\rm w}^2$. The following proposition establishes the monotonicity required to express the DEP requirement as an equivalent SNR constraint.

\begin{proposition}\label{prop_dep_monotonicity}
	For any positive integer $L$, $\zeta_L^\star(\rho)$ is continuous on $[0,\infty)$ and strictly decreasing for $\rho>0$. Moreover, $\zeta_L^\star(0)=1$ and $\lim_{\rho\to\infty}\zeta_L^\star(\rho)=0$.
\end{proposition}

\begin{proof}
	See Appendix~\ref{app:dep_monotonicity}.
\end{proof}

Proposition~\ref{prop_dep_monotonicity} guarantees the existence and uniqueness of the threshold below. Therefore, for a prescribed covertness parameter $\epsilon\in(0,1)$, the finite-observation-length covertness constraint
\begin{equation}
	\zeta_L^\star(\rho)\ge 1-\epsilon
\end{equation}
is equivalent to
\begin{equation}
	\rho\le \rho_{\epsilon,L},
\end{equation}
where $\rho_{\epsilon,L}$ is the unique positive solution of
\begin{equation}
	\zeta_L^\star(\rho_{\epsilon,L})=1-\epsilon.
	\label{rho_epsilon_def}
\end{equation}
Consequently, the exact finite-observation-length covertness constraint is equivalent to the deterministic received signal power constraint
\begin{equation}
	S_{\rm w}
	=
	{\bm q}_{\rm w}^H(\bm\xi)\bm W{\bm q}_{\rm w}(\bm\xi)
	\le
	\tau_{\epsilon,L},
	\label{exact_power_constraint}
\end{equation}
where
\begin{equation}
	\tau_{\epsilon,L}
	\triangleq
	\sigma_{\rm w}^2\rho_{\epsilon,L}.
\end{equation}
Although $\rho_{\epsilon,L}$ does not generally admit a closed-form expression, it can be computed efficiently by a one-dimensional bisection search using Eq.~\eqref{DEP_exact}. Hence, $\tau_{\epsilon,L}$ is a known scalar for fixed $(L,\epsilon,\sigma_{\rm w}^2)$, and the numerical inversion does not alter the structure of the subsequent joint optimization.

\subsection{Asymptotic Covertness Analysis for Infinite Observation Length}\label{sec:asymptotic_covertness}

We next characterize the limiting behavior for a fixed design as $L\to\infty$ under the assumed Gaussian signaling and known noise power. Specifically, $\bm W$, $\bm\xi$, and $\sigma_{\rm w}^2$, and hence $S_{\rm w}$ and $\rho=S_{\rm w}/\sigma_{\rm w}^2$, are held fixed as $L$ increases.

\begin{proposition}\label{prop_inf_covertness}
	For a fixed $S_{\rm w}$ under the considered model with known noise power, the minimum DEP at Willie satisfies
	\begin{equation}\label{DEP_inf}
		\lim_{L\to\infty}\zeta_L^\star(\rho)=
		\begin{cases}
			1, & S_{\rm w}=0,\\
			0, & S_{\rm w}>0.
		\end{cases}
	\end{equation}
	Therefore, for any prescribed covertness level $1-\epsilon$ with $\epsilon\in(0,1)$, the asymptotic covertness constraint
	\begin{equation}
		\lim_{L\to\infty}\zeta_L^\star(\rho) \ge 1-\epsilon
	\end{equation}
	holds if and only if
	\begin{equation}\label{asymp_constraint}
		S_{\rm w}=0.
	\end{equation}
\end{proposition}

\begin{proof}
By the strong law of large numbers, the statistic in Eq.~\eqref{radiometer} converges almost surely to
\begin{equation}\label{T_limit}
	T \xrightarrow{\mathrm{a.s.}}
	\begin{cases}
		\sigma_{\rm w}^2, & \mathcal H_0,\\
		\sigma_{\rm w}^2+S_{\rm w}, & \mathcal H_1.
	\end{cases}
\end{equation}
If $S_{\rm w}>0$, Willie can select any threshold satisfying $\sigma_{\rm w}^2<\eta<\sigma_{\rm w}^2+S_{\rm w}$. The two distinct limits in Eq.~\eqref{T_limit} then imply $P_{\rm FA}(\eta)\to0$ and $P_{\rm MD}(\eta)\to0$. Since the minimum DEP is nonnegative and cannot exceed the DEP attained by this threshold, $\zeta_L^\star(\rho)\to0$. If $S_{\rm w}=0$, the two hypotheses induce identical distributions for every $L$, and hence $\zeta_L^\star(0)=1$. Since $\epsilon\in(0,1)$, the asymptotic covertness constraint is therefore satisfied if and only if $S_{\rm w}=0$.
\end{proof}

Proposition~\ref{prop_inf_covertness} shows that, under the considered model, asymptotic covertness requires the zero leakage condition
\begin{equation}
{\bm q}_{\rm w}^H(\bm\xi)\bm W{\bm q}_{\rm w}(\bm\xi)=0.
\end{equation}
Polarforming therefore provides an additional DoF for reducing $S_{\rm w}$ beyond spatial covariance design alone. Whether zero leakage and a nonzero rate at Bob can be achieved simultaneously depends on the particular Alice-Bob and Alice-Willie channels.

\subsection{Problem Formulation}
For a prescribed finite observation length $L$, the joint design is formulated as ${\cal P}_0$ using the achievable rate in Eq.~\eqref{rate_bob} and the covertness requirement in Eq.~\eqref{covert_constraint}:
\begin{subequations}
\begin{align}\label{P1_covert}
\!\!\!\!\!\!	\mbox{${\cal P}_0$} \quad
	\max_{\bm\xi,\bm\phi,\bm W \succeq \bm 0} \quad
	& \log_2\det\left(
	\bm I_{N_\rmb}
\!+\!
	\tfrac{1}{\sigma_{\rm b}^2}{\bm H}_{\rm b}({\bm\xi},{\bm\phi})\bm W{\bm H}_{\rm b}^H({\bm\xi},{\bm\phi})
	\right)
	\\*
	\mathrm{s.t.}\qquad
	& \zeta_L^\star\!\left(\rho(\bm W,\bm\xi)\right) \ge 1-\epsilon, \label{P1_covertness}\\
	& \xi_{n_\rmt}\in\mathbb T,\ {n_\rmt}=1,\ldots,N_\rmt, \label{P1_xi}\\
	& \phi_{n_\rmb}\in\mathbb T,\ {n_\rmb}=1,\ldots,N_\rmb, \label{P1_phi}\\
	& \tr(\bm W)\le P. \label{P1_power}
\end{align}
\end{subequations}
Under the asymptotic model with a fixed design in Proposition~\ref{prop_inf_covertness}, the constraint in Eq.~\eqref{P1_covertness} reduces to $S_{\rm w}=0$. This limiting case is used only to characterize the asymptotic covertness requirement and is not solved by the proposed AO algorithm.
The exact DEP expression in Eq.~\eqref{DEP_exact} is not convenient for direct optimization. We therefore use the equivalent received signal power constraint in Eq.~\eqref{exact_power_constraint} and develop an AO algorithm for the resulting problem.

\section{Proposed Joint Optimization Algorithm}\label{sec:optimization}

This section develops an AO algorithm for the finite-observation-length design. We first reformulate the covertness constraint and then derive updates for the transmit covariance matrix, the receive PSV, and the transmit PSV. Finally, we summarize the overall algorithm and analyze its convergence and computational complexity.

For a prescribed finite observation length $L$, ${\cal P}_0$ is equivalently reformulated using Eq.~\eqref{exact_power_constraint} as
\begin{subequations}\label{P2_covert}
\begin{align}
\!\!\!
	\mbox{${\cal P}_1$} \quad
	\max_{\bm\xi,\bm\phi,\bm W \succeq \bm 0} \quad
	& \log_2\det\left(
	\bm I_{N_\rmb}
	+
	\frac{{\bm H}_{\rm b}({\bm\xi},{\bm\phi})\bm W{\bm H}_{\rm b}^H({\bm\xi},{\bm\phi})}{\sigma_{\rm b}^2}
	\right)
	\\*
	\mathrm{s.t.}\quad
	& {\bm q}_{\rm w}^H(\bm\xi)\bm W{\bm q}_{\rm w}(\bm\xi)\le \tau_{\epsilon,L},\label{P2_willie}\\
	& \text{Eqs.~}\eqref{P1_xi},~\eqref{P1_phi},~\eqref{P1_power}. \label{P2_inherited}
\end{align}
\end{subequations}
Although the DEP no longer appears explicitly, ${\cal P}_1$ remains non-convex because of the coupled covariance and phase variables. We address this coupling by cyclically optimizing $\bm W$, $\bm\phi$, and $\bm\xi$.

\subsection{Transmit Covariance Optimization}

For fixed $\bm\xi$ and $\bm\phi$, define
\begin{equation}
	\bar{\bm H}_{\rm b} \triangleq {\bm H}_{\rm b}({\bm\xi},{\bm\phi}),\qquad
	\bar{\bm q}_{\rm w} \triangleq {\bm q}_{\rm w}(\bm\xi).
\end{equation}
The covariance subproblem is formulated as
\begin{subequations}\label{PW}
\begin{align}
	\mbox{${\cal P}_2$} \quad
	\max_{\bm W \succeq \bm 0}\quad
	& \log_2\det\left(
	\bm I_{N_\rmb}
	+
	\frac{\bar{\bm H}_{\rm b}\bm W\bar{\bm H}_{\rm b}^H}{\sigma_{\rm b}^2}
	\right) \\
	\mathrm{s.t.}\quad
	& \text{Eqs.~}\eqref{P2_willie},~\eqref{P1_power}. \label{PW_constraints}
\end{align}
\end{subequations}
The objective is concave in $\bm W$, and both inequality constraints are affine in $\bm W$; hence, ${\cal P}_2$ is a convex optimization problem.
For $P>0$ and $\tau_{\epsilon,L}>0$, Slater's condition holds, and the Karush--Kuhn--Tucker (KKT) conditions are necessary and sufficient for optimality.

To characterize the optimal covariance, introduce dual variables $\lambda\ge0$ and $\mu\ge0$ for Eqs.~\eqref{P1_power} and \eqref{P2_willie}, respectively, and $\bm Z\succeq\bm 0$ for $\bm W\succeq\bm 0$. The Lagrangian is
\begin{align}
	\mathcal L(\bm W,\lambda,\mu,\bm Z)
	=
	&\log_2\det\left( \bm I_{N_\rmb}
	+
	\frac{\bar{\bm H}_{\rm b}\bm W\bar{\bm H}_{\rm b}^H}{\sigma_{\rm b}^2}
	\right) \notag\\
	&-\lambda(\tr(\bm W)-P)-\mu(\bar{\bm q}_{\rm w}^H\bm W\bar{\bm q}_{\rm w}-\tau_{\epsilon,L})\notag\\
	&
	+\tr(\bm Z\bm W).
\end{align}

Define
\begin{equation}
\bm K \triangleq \lambda \bm I+\mu \bar{\bm q}_{\rm w}\bar{\bm q}_{\rm w}^H.
\end{equation}
The KKT stationarity condition with respect to $\bm W$ is
\begin{equation}\label{KKT_stationarity}
	\frac{1}{\ln 2}\bar{\bm H}_{\rm b}^H
	\left(
	\sigma_{\rm b}^2\bm I_{N_\rmb}+\bar{\bm H}_{\rm b}\bm W\bar{\bm H}_{\rm b}^H
	\right)^{-1}
	\bar{\bm H}_{\rm b}
	=
	\bm K-\bm Z.
\end{equation}
We first consider the non-degenerate case $\lambda>0$, for which $\bm K\succ\bm 0$ and $\bm K^{-1/2}$ is well defined. A singular $\bm K$ can instead be handled by solving the original convex subproblem on the range of $\bm K$. For fixed $(\lambda,\mu)$, maximizing the Lagrangian over $\bm W$ is equivalent to
\begin{equation}
	\max_{\bm W\succeq \bm 0}\quad
	\log_2\det\left(
	\bm I_{N_\rmb}
	+
	\frac{\bar{\bm H}_{\rm b}\bm W\bar{\bm H}_{\rm b}^H}{\sigma_{\rm b}^2}
	\right)
	-\tr(\bm K\bm W).
\end{equation}

Let
\begin{equation}
	\tilde{\bm W}\triangleq \bm K^{1/2}\bm W\bm K^{1/2},
\end{equation}
and define the effective channel matrix
\begin{equation}
	\tilde{\bm H}\triangleq \bar{\bm H}_{\rm b}\bm K^{-1/2}.
\end{equation}
The transformed problem is
\begin{equation}
	\max_{\tilde{\bm W}\succeq \bm 0}\quad
	\log_2\det\left(
	\bm I_{N_\rmb}
	+
	\frac{1}{\sigma_{\rm b}^2}\tilde{\bm H}\tilde{\bm W}\tilde{\bm H}^H
	\right)
	-\tr(\tilde{\bm W}).
\end{equation}
Consider the eigenvalue decomposition
\begin{equation}
	\frac{1}{\sigma_{\rm b}^2}\bm K^{-1/2}\bar{\bm H}_{\rm b}^H\bar{\bm H}_{\rm b}\bm K^{-1/2}
	=
	\bm V\operatorname{diag}(\delta_1,\ldots,\delta_r)\bm V^H,
\end{equation}
where $r=\operatorname{rank}(\bar{\bm H}_{\rm b})$, $\delta_i$ is the $i$-th nonzero eigenvalue, and $\bm V\in\mathbb C^{N_\rmt\times r}$ contains the associated orthonormal eigenvectors. The optimal $\tilde{\bm W}$ shares these eigenvectors and is written as
\begin{equation}
	\tilde{\bm W}^\star
	=
	\bm V\operatorname{diag}(p_1,\ldots,p_r)\bm V^H,
\end{equation}
where the power coefficients solve
\begin{equation}
	\max_{p_i\ge 0}\ \sum_{i=1}^r \log_2(1+\delta_i p_i)-\sum_{i=1}^r p_i.
\end{equation}
The KKT conditions yield
\begin{equation}
p_i^\star=\left[\frac{1}{\ln 2}-\frac{1}{\delta_i}\right]^+,\qquad i=1,\ldots,r,
\end{equation}
Transforming back to the original coordinates gives the generalized water-filling structure
\begin{equation}\label{W_closed}
	\bm W^\star
	=
	\bm K^{-\frac{1}{2}}\bm V
	\operatorname{diag}(p_1^\star,\ldots,p_r^\star)
	\bm V^H\bm K^{-\frac{1}{2}}.
\end{equation}

At the dual optimum, primal feasibility and complementary slackness require
\begin{align}
	\lambda\big(\tr(\bm W^\star)-P\big)&=0,\\
	\mu\big(\bar{\bm q}_{\rm w}^H\bm W^\star \bar{\bm q}_{\rm w}-\tau_{\epsilon,L}\big)&=0.
\end{align}
The optimal dual variables can be obtained by minimizing the two-dimensional convex dual function, with Eq.~\eqref{W_closed} used to evaluate each candidate pair.

Therefore, for fixed $(\bm\xi,\bm\phi)$ and $\bm K\succ\bm0$, the optimal solution of ${\cal P}_2$ has the generalized water-filling structure in Eq.~\eqref{W_closed}, where $(\lambda,\mu)$ satisfy the KKT conditions.

\subsection{Receive Polarforming Optimization}
For fixed $\bm W$ and $\bm\xi$, the receive PSV $\bm\phi$ affects Bob's achievable rate but not Willie's received signal power. We therefore update one receive phase at a time while holding the remaining $N_\rmb-1$ phases fixed.

Specifically, let
\begin{equation}
	\bar{\bm P}(\bm\xi)\triangleq \bm P_{\rm b}\bm F(\bm\xi)\in\mathbb C^{2N_\rmb\times N_\rmt},
\end{equation}
and denote by
\begin{equation}
	\bar{\bm P}_n(\bm\xi)\triangleq
	\begin{bmatrix}
		\bm a_n^H\\
		\bm b_n^H
	\end{bmatrix}\in\mathbb C^{2\times N_\rmt}
\end{equation}
the $n$-th $2\times N_\rmt$ block of $\bar{\bm P}(\bm\xi)$, where $\bm a_n,\bm b_n\in\mathbb C^{N_\rmt\times 1}$. The $n$-th row of the effective Alice-Bob channel is
\begin{equation}
	\bm h_{{\rm b},n}^H(\phi_n)
	=
	\bm a_n^H+e^{-j\phi_n}\bm b_n^H.
\end{equation}

Let $\bm H_{{\rm b},-n}$ be $\bm H_{\rm b}(\bm\xi,\bm\phi)$ with its $n$-th row removed, and define
\begin{equation}
	\bm B_n \triangleq
	\bm I_{N_\rmt}
	+
	\frac{1}{\sigma_{\rm b}^2}\bm W^{\frac12}\bm H_{{\rm b},-n}^H\bm H_{{\rm b},-n}\bm W^{\frac12},
\end{equation}
which is independent of $\phi_n$. Sylvester's determinant theorem followed by the matrix determinant lemma gives
\begin{align}
	R_{\rm b}
	\!=\!
	\log_2\det(\bm B_n)
	+
	\log_2\left(
	1+
	\bm h_{{\rm b},n}^H(\phi_n)\bm T_n\bm h_{{\rm b},n}(\phi_n)
	\right),
\end{align}
where
\begin{equation}
	\bm T_n
	\triangleq
	\frac{1}{\sigma_{\rm b}^2}\bm W^{\frac12}\bm B_n^{-1}\bm W^{\frac12}.
\end{equation}
Therefore, optimizing $\phi_n$ is equivalent to
\begin{subequations}\label{Pphi_n}
\begin{align}
	\mbox{${\cal P}_{3,n}$}\quad
	\max_{\phi_n\in\mathbb T}\quad
	&\bm h_{{\rm b},n}^H(\phi_n)\bm T_n\bm h_{{\rm b},n}(\phi_n).
\end{align}
\end{subequations}
Substituting $\bm h_{{\rm b},n}(\phi_n)=\bm a_n+e^{j\phi_n}\bm b_n$ yields
\begin{align}
	\bm h_{{\rm b},n}^H(\phi_n)\bm T_n\bm h_{{\rm b},n}(\phi_n)
	=
	c_n+2\Re\left\{e^{j\phi_n}\beta_n\right\},
\end{align}
where
\begin{align}
	c_n &\triangleq \bm a_n^H\bm T_n\bm a_n+\bm b_n^H\bm T_n\bm b_n,\\
	\beta_n &\triangleq \bm a_n^H\bm T_n\bm b_n.
\end{align}
Let $\langle x\rangle_{2\pi}$ denote the unique representative of $x$ modulo $2\pi$ in $[0,2\pi)$. The optimal receive phase for the selected coordinate is
\begin{equation}\label{phi_closed_form}
	\phi_n^\star =
	\begin{cases}
		\langle-\angle(\beta_n)\rangle_{2\pi}, & \beta_n\neq 0,\\
		\text{arbitrary in } \mathbb T, & \beta_n=0.
	\end{cases}
\end{equation}

Successively applying Eq.~\eqref{phi_closed_form} for $n=1,\ldots,N_\rmb$ constitutes one cyclic coordinate-ascent sweep. Each coordinate update maximizes Bob's rate with respect to the selected phase and therefore cannot decrease the objective value.

\subsection{Transmit Polarforming Optimization}

For fixed $\bm W$ and $\bm\phi$, the transmit PSV $\bm\xi$ affects both Bob's achievable rate and Willie's received signal power. We again adopt cyclic coordinate ascent and optimize one phase $\xi_n$ at a time.

Define
\begin{equation}
	\hat{\bm P}(\bm\phi)
	\triangleq
	\bm G^H(\bm\phi)\bm P_{\rm b}
	\in\mathbb C^{N_\rmb\times 2N_\rmt},
\end{equation}
and let
\begin{equation}
	\hat{\bm P}_n(\bm\phi)
	\triangleq
	\begin{bmatrix}
		\bm c_n & \bm d_n
	\end{bmatrix}
	\in\mathbb C^{N_\rmb\times 2}
\end{equation}
denote the $n$-th $N_\rmb\times 2$ block of $\hat{\bm P}(\bm\phi)$, where $\bm c_n,\bm d_n\in\mathbb C^{N_\rmb\times 1}$. The $n$-th column of the effective Alice-Bob channel is
\begin{equation}
	\bm h_{{\rm b},n}(\xi_n)
	=
	\frac{1}{\sqrt2}
	\left(
	\bm c_n+e^{j\xi_n}\bm d_n
	\right).
	\label{hb_n_xi}
\end{equation}

Likewise, the $n$-th entry of the effective Alice-Willie channel is
\begin{equation}
	q_n(\xi_n)
	=
	\frac{1}{\sqrt2}
	\left(
	u_n+e^{j\xi_n}v_n
	\right),
	\label{qw_n_xi}
\end{equation}
where $u_n$ is the co-polar V-to-V coefficient and $v_n$ is the cross-polar H-to-V leakage coefficient associated with the $n$-th transmit PRA. Therefore,
\begin{equation}
	{\bm q}_{\rm w}^H(\bm\xi)
	=
	\begin{bmatrix}
		q_1(\xi_1) & \cdots & q_{N_\rmt}(\xi_{N_\rmt})
	\end{bmatrix}.
\end{equation}

Fix all phases except $\xi_n$, and let $\bm H_{{\rm b},-n}$ be $\bm H_{\rm b}(\bm\xi,\bm\phi)$ with its $n$-th column removed.
\begin{samepage}
After placing index $n$ last through a simultaneous row and column permutation, partition the covariance matrix as
\begin{equation}
	\bm W=
	\begin{bmatrix}
		\bm W_{-n,-n} & \bm w_{-n,n}\\
		\bm w_{-n,n}^H & w_{n,n}
	\end{bmatrix},
\end{equation}
\end{samepage}
where $w_{n,n}$ is the $n$-th diagonal entry of $\bm W$. The achievable rate is then
\begin{align}
	R_{\rm b}(\xi_n)
	&=
	\log_2\det\Big(
	\bm A_n
	+
	\frac{1}{\sigma_{\rm b}^2}
	\bm h_{{\rm b},n}(\xi_n)w_{n,n}\bm h_{{\rm b},n}^H(\xi_n)
	\notag\\
	&\qquad
	+
	\frac{1}{\sigma_{\rm b}^2}
	\bm h_{{\rm b},n}(\xi_n)\bm r_n^H
	+
	\frac{1}{\sigma_{\rm b}^2}
	\bm r_n\bm h_{{\rm b},n}^H(\xi_n)
	\Big),
	\label{Rb_xi_rank_two}
\end{align}
where
\begin{align}
	\bm A_n
	&\triangleq
	\bm I_{N_\rmb}
	+
	\frac{1}{\sigma_{\rm b}^2}
	\bm H_{{\rm b},-n}
	\bm W_{-n,-n}
	\bm H_{{\rm b},-n}^H,\\
	\bm r_n
	&\triangleq
	\bm H_{{\rm b},-n}\bm w_{-n,n}.
\end{align}
Since $\bm A_n\succ\bm0$ is independent of $\xi_n$, the dependence of $R_{\rm b}$ on this phase can be isolated. Define
\begin{equation}
	\alpha\triangleq\frac{1}{\sigma_{\rm b}^2},
	\qquad
	\bm C_n\triangleq\bm A_n^{-1},
\end{equation}
and
\begin{align}
	x_n(\xi_n)
	&\triangleq
	\bm h_{{\rm b},n}^H(\xi_n)
	\bm C_n
	\bm h_{{\rm b},n}(\xi_n),\\
	y_n(\xi_n)
	&\triangleq
	\bm h_{{\rm b},n}^H(\xi_n)
	\bm C_n
	\bm r_n,\\
	t_n
	&\triangleq
	\bm r_n^H\bm C_n\bm r_n.
\end{align}
Applying the rank-two matrix determinant lemma to Eq.~\eqref{Rb_xi_rank_two} yields
\begin{align}
	R_{\rm b}(\xi_n)
	&=
	\log_2\det(\bm A_n)
	+
	\log_2\Big(
	1
	+
	\alpha w_{n,n}x_n(\xi_n)
	\notag\\
	&\qquad\qquad
	+
	2\alpha\Re\{y_n(\xi_n)\}
	\notag\\
	&\qquad\qquad
	+
	\alpha^2
	\big(
	|y_n(\xi_n)|^2
	-
	x_n(\xi_n)t_n
	\big)
	\Big).
	\label{Rb_xi_detlemma}
\end{align}
Substituting Eq.~\eqref{hb_n_xi} into Eq.~\eqref{Rb_xi_detlemma}, the argument of the second logarithm reduces to
\begin{equation}
	\kappa_n
	+
	2\Re\left\{
	e^{j\xi_n}\theta_n
	\right\},
	\label{Rb_xi_affine_cosine}
\end{equation}
where $\kappa_n$ is independent of $\xi_n$ and is given by
\begin{align}
	\kappa_n
	\triangleq\;
	&1
	+
	\frac{\alpha w_{n,n}}{2}
	\left(
	\bm c_n^H\bm C_n\bm c_n
	+
	\bm d_n^H\bm C_n\bm d_n
	\right)
	\notag\\
	&+
	\sqrt{2}\alpha
	\Re\left\{
	\bm c_n^H\bm C_n\bm r_n
	\right\}
	\notag\\
	&+
	\frac{\alpha^2}{2}
	\left|\bm c_n^H\bm C_n\bm r_n\right|^2
	+
	\frac{\alpha^2}{2}
	\left|\bm d_n^H\bm C_n\bm r_n\right|^2
	\notag\\
	&-
	\frac{\alpha^2t_n}{2}
	\left(
	\bm c_n^H\bm C_n\bm c_n
	+
	\bm d_n^H\bm C_n\bm d_n
	\right),
	\label{kappa_n_def}
\end{align}
and
\begin{align}
	\theta_n
	\triangleq\;
	&\frac{1}{2}
	\left(
	\alpha w_{n,n}
	-
	\alpha^2 t_n
	\right)
	\bm c_n^H\bm C_n\bm d_n
	+
	\frac{\alpha}{\sqrt2}
	\bm r_n^H\bm C_n\bm d_n
	\notag\\
	&+
	\frac{\alpha^2}{2}
	\left(
	\bm c_n^H\bm C_n\bm r_n
	\right)
	\left(
	\bm r_n^H\bm C_n\bm d_n
	\right).
	\label{theta_n_def}
\end{align}
\begin{samepage}
Since the logarithm is strictly increasing, the unconstrained rate-maximizing phase is
\begin{equation}
	\xi_{n,{\rm rate}}^\star
	=
	\begin{cases}
		\langle-\angle(\theta_n)\rangle_{2\pi}, & \theta_n\neq 0,\\
		\xi_n^{\rm current}, & \theta_n=0.
	\end{cases}
	\label{xi_rate_opt_correct}
\end{equation}
\end{samepage}

Next, consider Willie's received signal power
\begin{equation}
	S_{\rm w}(\bm\xi,\bm W)
	=
	{\bm q}_{\rm w}^H(\bm\xi)\bm W{\bm q}_{\rm w}(\bm\xi).
\end{equation}
Using the same permutation that places index $n$ last, partition ${\bm q}_{\rm w}(\bm\xi)$ as
\begin{equation}
	{\bm q}_{\rm w}(\bm\xi)
	=
	\begin{bmatrix}
		{\bm q}_{{\rm w},-n}(\bm\xi_{-n})\\
		q_n^*(\xi_n)
	\end{bmatrix},
\end{equation}
and partition $\bm W$ accordingly as
\begin{equation}
	\bm W=
	\begin{bmatrix}
		\bm W_{-n,-n} & \bm w_{-n,n}\\
		\bm w_{-n,n}^H & w_{n,n}
	\end{bmatrix}.
\end{equation}
Then,
\begin{align}
	S_{\rm w}(\bm\xi,\bm W)
	&=
	{\bm q}_{{\rm w},-n}^H
	\bm W_{-n,-n}
	{\bm q}_{{\rm w},-n}
	+
	w_{n,n}|q_n(\xi_n)|^2 \notag\\
	&\quad
	+
	2\Re\left\{
	q_n^*(\xi_n)
	{\bm q}_{{\rm w},-n}^H\bm w_{-n,n}
	\right\}.
\end{align}
Define
\begin{align}
	\omega_n
	&\triangleq
	{\bm q}_{{\rm w},-n}^H
	\bm W_{-n,-n}
	{\bm q}_{{\rm w},-n},\\
	\ell_n
	&\triangleq
	{\bm q}_{{\rm w},-n}^H\bm w_{-n,n}.
\end{align}
Substituting Eq.~\eqref{qw_n_xi} into the above expression yields
\begin{equation}
	S_{\rm w}(\bm\xi,\bm W)
	=
	s_{0,n}
	+
	2\Re\left\{
	e^{j\xi_n}s_{1,n}
	\right\},
	\label{sw_xi_final}
\end{equation}
where $s_{0,n}$ is independent of $\xi_n$ and is given by
\begin{equation}
	s_{0,n}
	\triangleq
	\omega_n
	+
	\frac{w_{n,n}}{2}
	\left(
	|u_n|^2+|v_n|^2
	\right)
	+
	\sqrt2\Re\left\{
	u_n^*\ell_n
	\right\},
	\label{sw_constant_def}
\end{equation}
and
\begin{equation}
	s_{1,n}
	\triangleq
	\frac{w_{n,n}}{2}u_n^*v_n
	+
	\frac{1}{\sqrt2}v_n\ell_n^*.
	\label{sw_harmonic_def}
\end{equation}
In deriving Eq.~\eqref{sw_xi_final}, we used the identity
\begin{equation*}
	\Re\{e^{-j\xi_n}v_n^*\ell_n\}
	=
	\Re\{e^{j\xi_n}v_n\ell_n^*\}.
\end{equation*}

The coordinate subproblem is therefore
\begin{subequations}\label{Pxi_n}
\begin{align}
	\mbox{${\cal P}_{4,n}$}\quad
	\max_{\xi_n\in\mathbb T}\quad
	&
	\kappa_n
	+
	2\Re\left\{
	e^{j\xi_n}\theta_n
	\right\}
	\label{Pxi_n_obj}\\
	\mathrm{s.t.}\quad
	& \text{Eq.~}\eqref{P2_willie}.
	\label{Pxi_n_cons}
\end{align}
\end{subequations}
${\cal P}_{4,n}$ is a scalar constant-modulus optimization problem with a single trigonometric inequality constraint. If $\xi_{n,{\rm rate}}^\star$ in Eq.~\eqref{xi_rate_opt_correct} satisfies Eq.~\eqref{Pxi_n_cons}, then it is the optimal solution to ${\cal P}_{4,n}$. Otherwise, the optimal solution must lie on the boundary of Eq.~\eqref{Pxi_n_cons}, i.e.,
\begin{equation}
	S_{\rm w}(\bm\xi,\bm W)=\tau_{\epsilon,L}.
	\label{Pxi_boundary}
\end{equation}

When $s_{1,n}\neq 0$, writing $s_{1,n}=|s_{1,n}|e^{j\angle(s_{1,n})}$ and using Eq.~\eqref{sw_xi_final}, the boundary points satisfying Eq.~\eqref{Pxi_boundary} are
\begin{equation}
	\xi_n^{(\pm)}
	=
	\left\langle-\angle(s_{1,n})
	\pm
	\arccos\left(
	\frac{\tau_{\epsilon,L}-s_{0,n}}{2|s_{1,n}|}
	\right)\right\rangle_{2\pi},
	\label{xi_boundary_points}
\end{equation}
provided that
\begin{equation}
	\left|
	\tau_{\epsilon,L}-s_{0,n}
	\right|
	\le
	2|s_{1,n}|.
\end{equation}
The optimal phase is then selected as
\begin{equation}
	\xi_n^\star
	=
	\arg\max_{\xi_n\in\{\xi_n^{(+)},\xi_n^{(-)}\}}
	\left(
	\kappa_n
	+
	2\Re\left\{
	e^{j\xi_n}\theta_n
	\right\}
	\right).
	\label{xi_opt_boundary}
\end{equation}
If $s_{1,n}=0$, Willie's received signal power is independent of $\xi_n$, and the phase in Eq.~\eqref{xi_rate_opt_correct} that maximizes the rate is feasible whenever the current AO iterate is feasible. More generally, feasibility of the current iterate guarantees that the boundary points in Eq.~\eqref{xi_boundary_points} exist whenever the unconstrained phase that maximizes the rate is infeasible.

Successively solving Eq.~\eqref{Pxi_n} for $n=1,\ldots,N_\rmt$ constitutes one sweep over the transmit phases. Each update preserves feasibility and cannot decrease Bob's achievable rate.

\subsection{Overall Algorithm, Convergence, and Complexity}

Algorithm~\ref{algo_covert} summarizes the proposed AO procedure for ${\cal P}_0$, implemented through its equivalent reformulation ${\cal P}_1$. The initialization $\bm W^{(0)}=\bm0$ is feasible for any transmit and receive PSVs, provided that $P\ge0$ and $\tau_{\epsilon,L}\ge0$. Each covariance update solves ${\cal P}_2$ optimally for the current PSVs, and every phase update solves its scalar coordinate subproblem while preserving feasibility. Hence, no update decreases Bob's achievable rate. Because the transmit power is bounded and both phase domains are compact, the achievable rate is bounded above; consequently, the sequence of objective values converges. This monotonicity argument does not imply global optimality of the obtained solution.

\begin{algorithm}[!t]
	\caption{AO algorithm for ${\cal P}_0$}
	\label{algo_covert}
	\setstretch{0.86}
	\small
	\KwIn{$\bm P_{\rm b}$, $\bm h_{\rm w}^H$, $\sigma_{\rm b}^2$, $\sigma_{\rm w}^2$, $P$, $L$, $\epsilon$, $\varepsilon_{\rm AO}$, and $K_{\max}$}
	\KwOut{$\bm W$, $\bm\xi$, and $\bm\phi$}
	Compute $\rho_{\epsilon,L}$ from Eq.~\eqref{rho_epsilon_def} and set $\tau_{\epsilon,L}=\sigma_{\rm w}^2\rho_{\epsilon,L}$\;
	Choose $\bm\xi\in\mathbb T^{N_\rmt}$ and $\bm\phi\in\mathbb T^{N_\rmb}$; set $\bm W=\bm0$, $k=0$, and evaluate $R_{\rm b}^{(0)}$\;
	\Repeat{$\delta^{(k)}\le\varepsilon_{\rm AO}$ or $k=K_{\max}$}{
		Set $k\leftarrow k+1$\;
		Update $\bm W$ by solving ${\cal P}_2$ using Eq.~\eqref{W_closed}\;
		\For{$n_\rmb=1$ \KwTo $N_\rmb$}{
			Update $\phi_{n_\rmb}$ using Eq.~\eqref{phi_closed_form}\;
		}
		\For{$n_\rmt=1$ \KwTo $N_\rmt$}{
			Compute $\xi_{n_\rmt,{\rm rate}}^\star$ using Eq.~\eqref{xi_rate_opt_correct}\;
			\eIf{$\xi_{n_\rmt,{\rm rate}}^\star$ satisfies Eq.~\eqref{P2_willie}}{
				Set $\xi_{n_\rmt}\leftarrow\xi_{n_\rmt,{\rm rate}}^\star$\;
			}{
				Update $\xi_{n_\rmt}$ using Eq.~\eqref{xi_opt_boundary}\;
			}
		}
		Evaluate $R_{\rm b}^{(k)}$ and set $\delta^{(k)}=|R_{\rm b}^{(k)}-R_{\rm b}^{(k-1)}|/\max\{1,R_{\rm b}^{(k-1)}\}$\;
	}
	\Return{$\bm W$, $\bm\xi$, and $\bm\phi$}\;
\end{algorithm}

Let $I_{\rm d}$ denote the number of dual function evaluations required by the two-dimensional search in the covariance update. Using direct matrix decompositions, the dominant complexity in each iteration is $\mathcal O(I_{\rm d}N_\rmt^3)$ for updating $\bm W$, $\mathcal O(N_\rmb N_\rmt^3)$ for one sweep over the receive phases, and $\mathcal O(N_\rmt N_\rmb^3)$ for one sweep over the transmit phases. Thus, for $K_{\rm AO}$ outer iterations, the overall complexity is
\begin{equation}
\mathcal O\!\left(K_{\rm AO}\left[I_{\rm d}N_\rmt^3+N_\rmb N_\rmt^3+N_\rmt N_\rmb^3\right]\right),
\end{equation}
under a direct implementation. Matrix inverse updates can further reduce this computational cost.

\section{Numerical Results}\label{sec:numerical_results}

This section first validates the exact finite-observation-length DEP analysis and then evaluates the convergence of the proposed AO algorithm and the achievable rate of the joint polarforming design. The results also quantify the effects of the covertness requirement, Willie's observation length, and the polarized channel parameters. Unless otherwise stated, the default simulation parameters are summarized in Table~\ref{tab:simulation_parameters}.

\begin{table}[!t]
\centering
\caption{Default Simulation Parameters}
\label{tab:simulation_parameters}
\footnotesize
\setlength{\tabcolsep}{3pt}
\renewcommand{\arraystretch}{0.96}
\begin{tabular}{@{}>{\raggedright\arraybackslash}p{0.62\columnwidth}
                    >{\centering\arraybackslash}p{0.32\columnwidth}@{}}
\hline
Parameter & Default value \\
\hline
Numbers of PRAs, $(N_\rmt,N_\rmb)$
    & $(2,2)$ \cite{zhou2024polarforming} \\
Maximum transmit power, $P$ & $10$~dBm \\
Observation length, $L$ & $200$ samples \\
Covertness parameter, $\epsilon$ & $0.05$ \cite{11603394} \\
Noise powers, $(\sigma_{\rm b}^2,\sigma_{\rm w}^2)$ & $(-90,-90)$~dBm \\
Antenna inverse XPIs, $(\mu_{\rm t},\mu_{\rm b})$ & $(0.1,0.1)$ \\
Channel inverse XPRs, $(\chi_{\rm b},\chi_{\rm w})$
    & $(0.2,0.1)$ \cite{11235492,11603394} \\
Average Alice-Bob channel power gain, $\beta_{\rm b}$ & $-90$~dB \\
Total average Alice-Willie channel power gain,
    $\beta_{{\rm w},V}+\beta_{{\rm w},H}$ & $-90$~dB \\
Monte Carlo channel realizations & $100$ \\
\hline
\end{tabular}
\end{table}

All transmit and noise powers are expressed in watts in the simulations and reported in dBm, whereas $\beta_{\rm b}$, $\beta_{{\rm w},V}$, and $\beta_{{\rm w},H}$ are dimensionless channel power gains reported in dB. All channels are independently generated according to the Rayleigh fading models in Section~\ref{sec:system_model}. For the Alice-Willie link, $u_n\sim\mathcal{CN}(0,\beta_{{\rm w},V})$ and $v_n\sim\mathcal{CN}(0,\beta_{{\rm w},H})$, with $10\log_{10}(\beta_{{\rm w},V}+\beta_{{\rm w},H})=-90$~dB maintained as $\chi_{\rm w}$ varies.

We compare the proposed joint design with the following benchmarks \cite{11591198}:
\begin{itemize}
    \item \emph{Transmit polarforming only (TPO):} Fix $\bm\phi=\bm0$ and optimize $(\bm W,\bm\xi)$.
    \item \emph{Fixed zero phase (FZP):} Set $\bm\xi=\bm0$ and $\bm\phi=\bm0$, and optimize only $\bm W$.
    \item \emph{Vertical linearly polarized antenna (LPA):} Fix the polarization vector to $[1,0]^T$ at Alice and Bob, and optimize only $\bm W$.
    \item \emph{Left-hand circularly polarized antenna (CPA):} Fix the polarization vectors to $[1,j]^T/\sqrt{2}$ at Alice and $[1,j]^T$ at Bob, and optimize only $\bm W$.
    \item \emph{Random phase (RP):} Draw the transmit and receive phases independently and uniformly over $[0,2\pi)$, and optimize only $\bm W$.
    \item \emph{Achievable-rate upper bound (UB):} Remove Eq.~\eqref{P2_willie} and jointly optimize all variables.
\end{itemize}

\subsection{Validation of the Exact Finite-Observation-Length DEP}
Fig.~\ref{fig:dep_snr} compares the exact minimum DEP in Eq.~\eqref{DEP_exact} with Monte Carlo results for different observation lengths. The close agreement validates the finite-observation-length DEP characterization. The DEP decreases as Willie's received SNR increases, confirming that stronger signal leakage makes the transmission easier to detect. Moreover, for a given received SNR, a longer observation length yields a lower DEP because Willie can accumulate more evidence for distinguishing the two hypotheses. These results corroborate Proposition~\ref{prop_dep_monotonicity} and show that the allowable signal leakage depends critically on Willie's observation length, which must therefore be explicitly incorporated into the covertness constraint.

\begin{figure}[!t]
	\centering
	\includegraphics[width=\columnwidth]{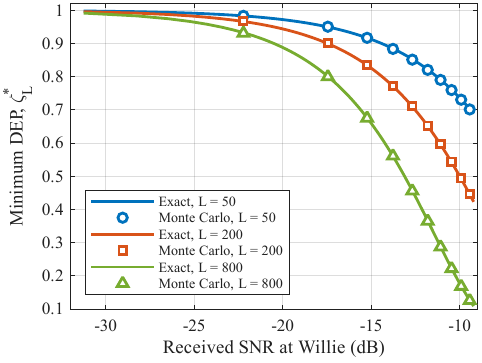}
	\caption{Exact minimum DEP and Monte Carlo results versus the received SNR for different observation lengths in samples.}
	\label{fig:dep_snr}
\end{figure}

\begin{figure}[!t]
	\centering\vspace{-4mm}
	\includegraphics[width=\columnwidth]{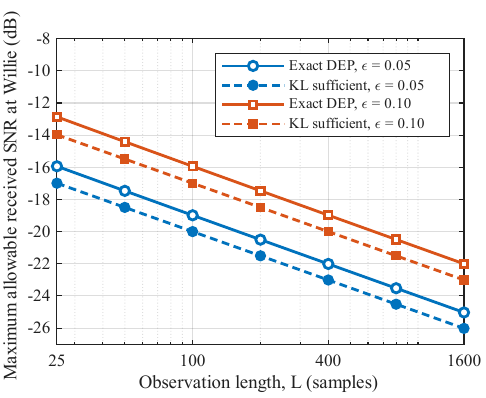}
\caption{Maximum allowable received SNR at Willie versus his observation length in samples under the exact DEP constraint and the sufficient Pinsker/Kullback--Leibler (KL) condition.}
	\label{fig:admissible_snr}
\end{figure}

In Fig.~\ref{fig:admissible_snr}, the maximum allowable received SNR at Willie is the largest $\rho$ satisfying $\zeta_L^\star(\rho)\ge 1-\epsilon$. It decreases with $L$ under both the exact DEP constraint and the sufficient KL condition $\rho\le 2\epsilon/\sqrt{L}$ adopted in \cite{11603394}. The exact threshold is $0.99$--$1.12$~dB higher for $\epsilon\in\{0.05,0.1\}$, showing that the sufficient condition unnecessarily reduces the feasible received signal power and that directly imposing the exact DEP constraint avoids this conservatism.

\subsection{Algorithm Convergence and Rate Comparison}
In Fig.~\ref{fig:ao_convergence}, Channels 1--4 denote four independently generated realizations of the polarized Alice--Bob and Alice--Willie channels under the same default parameters. The achievable rate increases monotonically and converges within eight iterations for all four realizations, demonstrating the rapid convergence of the proposed AO algorithm.

\begin{figure}[!t]
	\centering
	\includegraphics[width=\columnwidth]{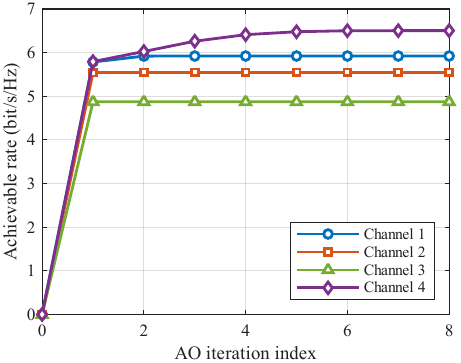}
	\caption{Convergence of the proposed AO algorithm for four representative polarized channel realizations.}
	\label{fig:ao_convergence}
\end{figure}

Fig.~\ref{fig:rate_power} compares the average achievable rates of the joint polarforming design and the benchmark schemes. At $P=10$~dBm, the joint design outperforms TPO by $8.8\%$, demonstrating the additional rate gain provided by adapting Bob's receive PSV. Its $21.0\%$ gain over FZP shows the benefit of jointly optimizing the transmit and receive PSVs rather than relying only on transmit covariance optimization with fixed phases. The gains of $53.8\%$ and $45.3\%$ over LPA and CPA, respectively, indicate that fixed polarization states cannot adapt to the polarized Alice--Bob channel while controlling the leakage to Willie. The $37.5\%$ gain over RP further confirms that the phase shifts must be deliberately coordinated with the transmit covariance rather than randomly selected. Across the considered power range, the joint design consistently outperforms all covert benchmarks, while its gap to UB widens as $P$ increases. This trend indicates that the covertness constraint becomes increasingly restrictive at high transmit power, yet joint spatial and polarization adaptation uses the available power more effectively than the benchmark designs.

\begin{figure}[!t]
	\centering
	\includegraphics[width=\columnwidth]{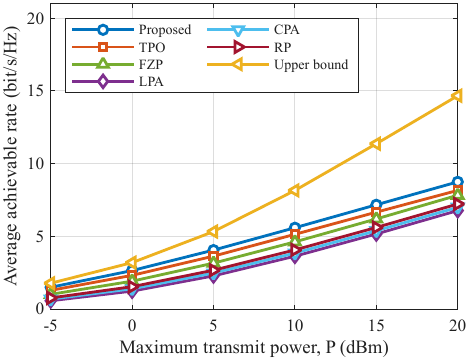}
	\caption{Average achievable rate versus the maximum transmit power for the joint polarforming design and benchmark schemes.}
	\label{fig:rate_power}
\end{figure}

\subsection{Rate and Covertness Tradeoffs}
Fig.~\ref{fig:rate_tradeoffs}(a) shows that relaxing the covertness requirement by increasing $\epsilon$ raises the maximum allowable received signal power at Willie and thereby improves Bob's average achievable rate. The diminishing rate increase indicates that the Alice--Bob channel and the maximum transmit power also limit performance after the covertness constraint is relaxed.

Fig.~\ref{fig:rate_tradeoffs}(b) shows that Bob's average achievable rate decreases as Willie's observation length increases. A longer observation interval strengthens Willie's detector and lowers the maximum allowable received signal power at Willie. The modest rate reduction indicates that joint spatial and polarization adaptation accommodates much of the tightened covertness requirement by redirecting the transmitted signal.

\begin{figure}[!t]
	\centering
	\includegraphics[width=\columnwidth]{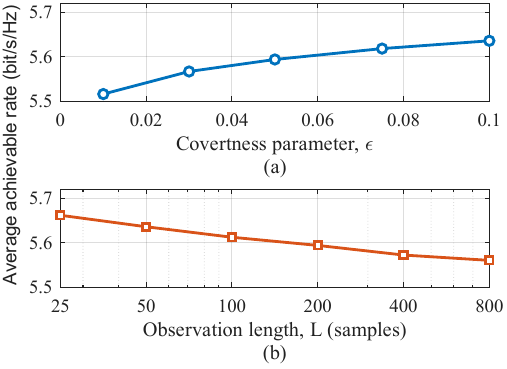}
	\caption{Average achievable rate versus (a) the covertness parameter and (b) Willie's observation length.}
	\label{fig:rate_tradeoffs}
\end{figure}

\subsection{Impact of Polarization Coupling and Cross-Polar Leakage}
Fig.~\ref{fig:rate_chib} shows that the joint design varies by only $1.3\%$ over the considered inverse XPR range, whereas the LPA rate decreases by $19.4\%$ as $\chi_{\rm b}$ increases. This contrast indicates that jointly adapting the transmit and receive PSVs compensates for changes in the co-polar and cross-polar channel components. Adaptive polarforming is therefore particularly beneficial when the degree of channel depolarization is strong or uncertain, while a fixed linear polarization may incur substantial performance loss.

\begin{figure}[!t]
	\centering
	\includegraphics[width=\columnwidth]{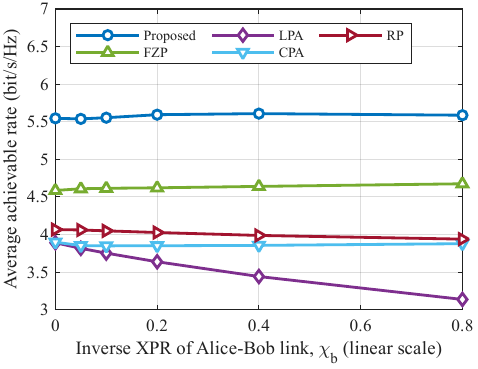}
	\caption{Average achievable rate versus the inverse XPR of the Alice-Bob link on a linear scale.}
	\label{fig:rate_chib}
\end{figure}

Fig.~\ref{fig:rate_chiw} keeps the total Alice--Willie channel power gain fixed and redistributes it between the co-polar and cross-polar components through $\chi_{\rm w}$. Introducing cross-polar leakage initially reduces the joint-design rate by $3.5\%$ because the transmit PSV then affects both Bob's channel and Willie's received signal power; over the remaining positive range, however, the rate varies by less than $0.5\%$. At $\chi_{\rm w}=0.1$, the joint design retains gains of approximately $35\%$--$54\%$ over RP, CPA, and LPA. Thus, when cross-polar leakage is nonnegligible, the transmit covariance and both PSVs should be optimized jointly rather than selecting the polarization state independently of the covertness constraint.

\begin{figure}[!t]
	\centering
	\includegraphics[width=\columnwidth]{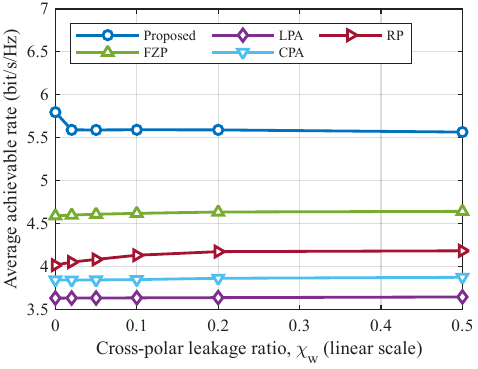}
	\caption{Average achievable rate versus the Alice-Willie cross-polar leakage ratio on a linear scale for a fixed total average channel power gain.}
	\label{fig:rate_chiw}
\end{figure}

\FloatBarrier
\section{Conclusion}\label{sec:conclusion}
This paper investigated PRA-aided MIMO covert communication by jointly exploiting spatial and polarization DoFs. We derived Willie's optimal radiometric threshold and exact minimum DEP for a finite observation length, transformed the covertness requirement into an equivalent received signal power constraint, and characterized the zero-leakage requirement for asymptotic covertness under a fixed transmission design. For the finite-observation-length design, we developed an AO algorithm combining a generalized water-filling covariance update with closed-form transmit and receive PSV updates. Numerical results validated the analysis and demonstrated that joint polarforming improves achievable rate and adapts effectively to changes in the covertness requirement and polarized channel composition. Future work will consider imperfect polarized channel knowledge, noise uncertainty, and wardens with adaptive polarization or multiple antennas. Network-level extensions to large-scale UAV systems may further integrate polarforming-enabled covert links with edge server placement and lightweight consensus mechanisms \cite{11146623,11180147}.

\appendices
\section{Proof of Proposition~\ref{prop_dep_monotonicity}}\label{app:dep_monotonicity}
Let $X\triangleq LT/\sigma_{\rm w}^2$. Under $\mathcal H_0$ and $\mathcal H_1$, the densities of $X$ are, respectively,
\begin{align}
	f_0(x)
	&=
	\frac{x^{L-1}e^{-x}}{\Gamma(L)},\\
	f_1(x;\rho)
	&=
	\frac{x^{L-1}e^{-x/(1+\rho)}}{\Gamma(L)(1+\rho)^L},
	\qquad x\ge0.
\end{align}
For $\rho>0$, the likelihood ratio $f_1(x;\rho)/f_0(x)$ is strictly increasing in $x$ and equals one at
\begin{equation}
	a(\rho)
	\triangleq
	L\frac{1+\rho}{\rho}\ln(1+\rho).
\end{equation}
Define $b(\rho)\triangleq a(\rho)/(1+\rho)=L\ln(1+\rho)/\rho$ and
\begin{equation}
	g_L(x)\triangleq\frac{x^{L-1}e^{-x}}{\Gamma(L)}.
\end{equation}
The exact minimum DEP in Eq.~\eqref{DEP_exact} can then be written as
\begin{equation}
	\zeta_L^\star(\rho)
	=
	\frac{\Gamma(L,a(\rho))}{\Gamma(L)}
	+
	\frac{\gamma(L,b(\rho))}{\Gamma(L)}.
\end{equation}
Differentiating with respect to $\rho$ gives
\begin{equation}\label{dep_derivative_step}
	\frac{\mathrm d\zeta_L^\star(\rho)}{\mathrm d\rho}
	=
	-g_L(a(\rho))a'(\rho)
	+g_L(b(\rho))b'(\rho).
\end{equation}
At the unique density intersection, $f_0(a(\rho))=f_1(a(\rho);\rho)$, which is equivalent to
\begin{equation}
	g_L(b(\rho))=(1+\rho)g_L(a(\rho)).
\end{equation}
Moreover, differentiating $b(\rho)=a(\rho)/(1+\rho)$ yields
\begin{equation}
	(1+\rho)b'(\rho)
	=
	a'(\rho)-\frac{a(\rho)}{1+\rho}.
\end{equation}
Substituting these two identities into Eq.~\eqref{dep_derivative_step} gives
\begin{equation}
	\frac{\mathrm d\zeta_L^\star(\rho)}{\mathrm d\rho}
	=
	-\frac{a(\rho)}{1+\rho}g_L(a(\rho))<0,
	\qquad \rho>0.
\end{equation}
Thus, $\zeta_L^\star(\rho)$ is strictly decreasing for $\rho>0$. Since $a(\rho)\to L$ and $b(\rho)\to L$ as $\rho\to0$, the identity $\Gamma(L,L)+\gamma(L,L)=\Gamma(L)$ gives $\lim_{\rho\to0}\zeta_L^\star(\rho)=1$, which agrees with the definition at $\rho=0$. Finally, $a(\rho)\to\infty$ and $b(\rho)\to0$ as $\rho\to\infty$, yielding $\lim_{\rho\to\infty}\zeta_L^\star(\rho)=0$. This also establishes continuity on $[0,\infty)$.

\balance
\bibliographystyle{IEEEtran}
\bibliography{IEEEref}

@ARTICLE{11207211,
  author={Cheng, Xiang and Zhang, Jianan and others},
  journal={J. Commun. Inf. Netw.}, 
  title={Integrated {AI} and Communications: A Two-Way Catalysis Toward {6G} and Beyond}, 
  year={Sept. 2025},
  volume={10},
  number={3},
  pages={191-200},
  doi={10.23919/JCIN.2025.11207211}}

@ARTICLE{11603394,
  author={Wu, Yang and Ma, Ruiqian and Lin, Zhi and Zhang, Peng and Wang, Guohua and Yang, Weiwei},
  journal={IEEE Wireless Commun. Lett.},
  title={Polarforming Design for Covert Communication},
  year={2026},
  note={Early access, doi: 10.1109/LWC.2026.3712221},
  doi={10.1109/LWC.2026.3712221}}

@article{xu2026sagsin_outage_capacity,
	author  = {Jinpeng Xu and Yingqi He and Lin Zhou and Jingjing Wang and Jun Du and Chunxiao Jiang},
	title   = {{6G} Space-Air-Ground-Sea Integrated Networks: Outage and Ergodic Capacity Analysis},
	journal = {IEEE Trans. Wireless Commun.},
	year={2026},
	volume={25},
	number={},
	pages={19692-19708}
}

@ARTICLE{11165349,
  author={Shao, Xiaodan and Zhang, Rui and Zhou, Haibo and Jiang, Qijun and Zhou, Conghao and Zhuang, Weihua and Shen, Xuemin},
  journal={IEEE J. Sel. Areas Commun.},
  title={Polarforming Antenna Enhanced Sensing and Communication: Modeling and Optimization},
  year={2026},
  volume={44},
  number={},
  pages={416-431},
  doi={10.1109/JSAC.2025.3610487}}

@article{xu2026digital,
 	title={Digital Twin Meets World Model: Proactive Control for {UAV} Swarms in Low-Altitude Wireless Networks},
 	author={Xu, Jinpeng and Xie, Xin and Chen, Qun and Mao, Jinfeng and Zhang, Wei},
 	journal={IEEE Trans. Netw. Sci. Eng.},
 	year={2026},
 	volume={},
 	number={},
 	pages={1-19},
 	doi={10.1109/TNSE.2026.3731054}
 }

@ARTICLE{11558360,
  author={Li, Haoyue and You, Li and Qian, Mengyu and Wang, Yuanshuo and Zhou, Huibin and Gao, Xiqi},
  journal={IEEE Wireless Commun. Lett.},
  title={Polarforming and Precoding Design for Massive {MIMO} {LEO} Satellite Communications},
  year={2026},
  note={Early access, doi: 10.1109/LWC.2026.3702745},
  doi={10.1109/LWC.2026.3702745}}

@ARTICLE{11235492,
  author={Ding, Jingze and Zhou, Zijian and Jiao, Bingli and Zhang, Rui},
  journal={IEEE Wireless Commun. Lett.},
  title={Secure Wireless Communication via Polarforming},
  year={2026},
  volume={15},
  number={},
  pages={480-484},
  doi={10.1109/LWC.2025.3630216}}

@ARTICLE{11547205,
  author={Mehrzad, Saba and Ebrahimi, Kamran and Attaran, Ali and Fang, Fang},
  journal={IEEE Wireless Commun. Lett.},
  title={Polarforming for {IRS}-Assisted Flexible Intelligent Metasurfaces-Enabled Wireless Communications},
  year={2026},
  volume={15},
  number={},
  pages={3576-3580},
  doi={10.1109/LWC.2026.3699597}}

@ARTICLE{11206360,
  author={Zhou, Zijian and Ding, Jingze and Zhang, Rui},
  journal={IEEE Wireless Commun. Lett.},
  title={Polarforming Design for Movable Antenna Systems},
  year={2026},
  volume={15},
  number={},
  pages={161-165},
  doi={10.1109/LWC.2025.3622756}}

@ARTICLE{11357497,
  author={Zhou, Zijian and Ding, Jingze and Xu, Jie and Zhang, Rui},
  journal={J. Commun. Inf. Netw.},
  title={Modeling and Optimization of Three-Dimensional Polarforming for Multiuser Communications},
  year={2025},
  month={Dec.},
  volume={10},
  number={4},
  pages={364-375},
  doi={10.23919/JCIN.2025.11357497}}

@ARTICLE{11347516,
  author={Ding, Jingze and Zhou, Zijian and Shao, Xiaodan and Jiao, Bingli and Zhang, Rui},
  journal={IEEE Commun. Mag.},
  title={Polarforming for Wireless Networks: Opportunities and Challenges},
  year={2026},
  month={Mar.},
  volume={64},
  number={3},
  pages={118-124},
  doi={10.1109/MCOM.001.2500341}}

@ARTICLE{Oestges_TWC_2008,
	author={Oestges, Claude and Clerckx, Bruno and Guillaud, Maxime and Debbah, Merouane},
	journal={IEEE Trans. Wireless Commun.},
	title={Dual-polarized wireless communications: {From} propagation models to system performance evaluation},
	year={2008},
	volume={7},
	month={Oct.},
	number={10},
	pages={4019-4031}
}

@misc{zhou2024polarforming,
	title={Polarforming for Wireless Communications: Modeling and Performance Analysis},
	author={Zhou, Zijian and Ding, Jingze and Wang, Chenbo and Jiao, Bingli and Zhang, Rui},
	year={2025},
	note={[Online]. Available: https://arxiv.org/abs/2409.07771},
}

@ARTICLE{10923651,
  author={Li, Xingwang and Liu, Musen and Dang, Shuping and Luong, Nguyen Cong and Yuen, Chau and Nallanathan, Arumugam and Niyato, Dusit},
  journal={IEEE Trans. Wireless Commun.},
  title={Covert Communications With Enhanced Physical Layer Security in {RIS}-Assisted Cooperative Networks},
  year={2025},
  month={Jul.},
  volume={24},
  number={7},
  pages={5605-5619},
  doi={10.1109/TWC.2025.3548024}}

@ARTICLE{11214184,
  author={Zhao, Xiang and Lu, Wencong and Yi, Changyan and Wang, Junyi and Peng, Jie},
  journal={IEEE Trans. Commun.},
  title={Multi-{UAV} Covert Communication With Informed Jammers: Design, Analysis, and Optimization},
  year={2026},
  volume={74},
  number={},
  pages={565-580},
  doi={10.1109/TCOMM.2025.3624167}}

@ARTICLE{10521454,
  author={Che, Bohan and Shi, Hui and Lu, Xingbo and Gao, Chang and Yang, Weiwei and Liu, Jue and Jiang, Junhao},
  journal={IEEE Trans. Veh. Technol.},
  title={Covert Multi-Channel Communication Against Interference-Assisted Proactive Detection},
  year={2024},
  month={Sep.},
  volume={73},
  number={9},
  pages={14033-14037},
  doi={10.1109/TVT.2024.3394734}}

@ARTICLE{11417854,
  author={Yang, Hanyu and Gong, Shiqi and Liu, Heng and Yu, Tao and Xing, Chengwen},
  journal={IEEE J. Sel. Areas Commun.},
  title={{MA}-Aided Integrated Sensing and Covert Communication Systems},
  year={2026},
  volume={44},
  number={},
  pages={3690-3704},
  doi={10.1109/JSAC.2026.3669141}}

@ARTICLE{11503892,
  author={Chen, Hongchi and Mu, Junsheng and Deng, Na and Wei, Haichao and Zhao, Nan},
  journal={IEEE Trans. Wireless Commun.},
  title={Air-to-Ground Covert Communication With Location and Interference Uncertainty},
  year={2026},
  volume={25},
  number={},
  pages={16187-16199},
  doi={10.1109/TWC.2026.3687670}}

@ARTICLE{10185203,
  author={Ma, Ruiqian and Yang, Weiwei and Guan, Xinrong and Lu, Xingbo and Song, Yi and Chen, Dechuan},
  journal={IEEE Internet Things J.},
  title={Covert {mmWave} Communications With Finite Blocklength Against Spatially Random Wardens},
  year={2024},
  month={Jan.},
  volume={11},
  number={2},
  pages={3402-3416},
  doi={10.1109/JIOT.2023.3296414}}

@ARTICLE{11122485,
  author={Li, Yueran and Deng, Na and Xing, Chengwen and Wei, Haichao and Zhao, Nan and Niyato, Dusit},
  journal={IEEE Trans. Commun.},
  title={Covert {UAV} Communication With Interference Uncertainty},
  year={2025},
  month={Dec.},
  volume={73},
  number={12},
  pages={15324-15336},
  doi={10.1109/TCOMM.2025.3597648}}

@ARTICLE{10244106,
  author={Li, Lintao and Chen, Zhilin and Chen, Riqing and Yang, Lei and Yan, Shihao},
  journal={IEEE Trans. Veh. Technol.},
  title={Covert Wireless Communication With Random Frequency Diverse Array},
  year={2024},
  month={Jan.},
  volume={73},
  number={1},
  pages={1473-1478},
  doi={10.1109/TVT.2023.3310986}}

@ARTICLE{10440606,
  author={Jiang, Yu'e and Wang, Liangmin and Chen, Hsiao-Hwa and Shen, Xuemin},
  journal={Proc. IEEE},
  title={Physical Layer Covert Communication in {B5G} Wireless Networks---Its Research, Applications, and Challenges},
  year={2024},
  month={Jan.},
  volume={112},
  number={1},
  pages={47-82},
  doi={10.1109/JPROC.2024.3364256}}

@ARTICLE{10879382,
  author={Zhao, Shuangrui and Liu, Jia and Shen, Yulong and Jiang, Xiaohong and Taleb, Tarik and Shiratori, Norio},
  journal={IEEE Trans. Inf. Forensics Security},
  title={On the Impact of Warden Collusion on Covert Communication in Wireless Networks},
  year={2025},
  volume={20},
  number={},
  pages={2297-2312},
  doi={10.1109/TIFS.2025.3540575}}

@ARTICLE{11313487,
  author={Liu, Mingqian and Wen, Zhaoxi and Chen, Yunfei and Tang, Jie and Wong, Kai-Kit and Yang, Xiaoniu},
  journal={IEEE J. Sel. Areas Commun.},
  title={Beamforming-Enabled Covert Communications for Multi-Position Warden},
  year={2026},
  volume={44},
  number={},
  pages={1722-1735},
  doi={10.1109/JSAC.2025.3647611}}

@ARTICLE{11215610,
  author={Chen, Chunqi and Wang, Manlin and Xu, Zhen and Lv, Xing and Xia, Bin},
  journal={IEEE Trans. Wireless Commun.},
  title={Performance Analysis and Optimization of {MIMO} Covert Communications With Finite-Alphabet Inputs},
  year={2026},
  volume={25},
  number={},
  pages={5922-5934},
  doi={10.1109/TWC.2025.3621695}}

@ARTICLE{10086651,
  author={Yang, Liang and Zhang, Wei and Bithas, Petros S. and Liu, Hongwu and Hasna, Mazen O. and Tsiftsis, Theodoros A. and Ng, Derrick Wing Kwan},
  journal={IEEE Trans. Veh. Technol.},
  title={Covert Transmission and Secrecy Analysis of {RS-RIS-NOMA}-Aided {6G} Wireless Communication Systems},
  year={2023},
  month={Aug.},
  volume={72},
  number={8},
  pages={10659-10670},
  doi={10.1109/TVT.2023.3262943}}

@ARTICLE{11141469,
  author={Qi, Peilin and Guo, Kefeng and Nauman, Ali and Wu, Qihui and Zhang, Lei and Wu, Zeke and Singh, Keshav},
  journal={IEEE Internet Things J.},
  title={Reliable Covert Communication in {NOMA}-Aided Cognitive Satellite Aerial Terrestrial Integrated Networks},
  year={2026},
  month={Mar.},
  volume={13},
  number={5},
  pages={7989-8002},
  doi={10.1109/JIOT.2025.3602642}}

@ARTICLE{10195191,
  author={Yang, Guo and Qian, Yuwen and Ren, Ke and Mei, Zhen and Shu, Feng and Zhou, Xiangwei and Wu, Wen},
  journal={IEEE J. Sel. Topics Signal Process.},
  title={Covert Wireless Communications for Augmented Reality Systems With Dual Cooperative {UAVs}},
  year={2023},
  month={Sep.},
  volume={17},
  number={5},
  pages={1119-1130},
  doi={10.1109/JSTSP.2023.3299116}}

@ARTICLE{11365966,
  author={Wang, Chao and Li, Zan and Zhou, Xiangnan and Jin, Liang},
  journal={IEEE J. Sel. Topics Signal Process.},
  title={Intelligent Cooperative Covert {ISAC} via Fluid Antennas: A {GAN}-Optimized Framework},
  year={2026},
  month={Apr.},
  volume={20},
  number={3},
  pages={262-276},
  doi={10.1109/JSTSP.2026.3658184}}

@ARTICLE{10970466,
  author={Chen, Bingwen and Ling, Xintong and Cao, Weihang and Wang, Jiaheng and Ding, Zhi},
  journal={IEEE J. Sel. Areas Commun.},
  title={Analysis of Channel Uncertainty in Trusted Wireless Services via Repeated Interactions}, 
  year={2025},
  month={Jun.},
  volume={43},
  number={6},
  pages={2248-2265},
  doi={10.1109/JSAC.2025.3560005}}

@ARTICLE{11164329,
  author={Gao, Pengpei and Wang, Zhining and Yu, Shimao and Song, Chunhe and Yu, Haibin},
  journal={IEEE Internet Things J.},
  title={Reliable and Trustworthy Local--Global Hierarchical Framework for Intrusion Detection in {6G-IoT} Networks},
  year={2026},
  month={Mar.},
  volume={13},
  number={5},
  pages={8077-8091},
  doi={10.1109/JIOT.2025.3609811}}

@ARTICLE{10844934,
  author={Jayasena, Aruna and Bachmann, Richard and Mishra, Prabhat},
  journal={IEEE Trans. Inf. Forensics Security},
  title={{CiseLeaks}: Information Leakage Assessment of Cryptographic Instruction Set Extension Prototypes},
  year={2025},
  volume={20},
  number={},
  pages={1551-1565},
  doi={10.1109/TIFS.2025.3531239}}

@ARTICLE{11534098,
  author={Wang, Jiaxing and Xu, Hengyan and Xie, Xin and Han, Rui and Bai, Lin and Shi, Guowei},
  journal={IEEE J. Sel. Topics Signal Process.},
  title={A Robust Physical Layer Security Scheme Based on Collaborative Beamforming for {UAV} Networks with Jitter},
  year={2026},
  note={Early access, doi: 10.1109/JSTSP.2026.3696183},
  doi={10.1109/JSTSP.2026.3696183}
}

@ARTICLE{10938035,
  author={Xu, Jinpeng and Bai, Lin and Xie, Xin and Zhou, Lin},
  journal={IEEE Trans. Wireless Commun.},
  title={Collaborative Secret and Covert Communications for Multi-User Multi-Antenna Uplink {UAV} Systems: Design and Optimization},
  year={2025},
  month={Jul.},
  volume={24},
  number={7},
  pages={6020-6035},
  doi={10.1109/TWC.2025.3550911}
}

@ARTICLE{10286081,
  author={Wang, Jingjing and Jiao, Zihan and Chen, Jianrui and Hou, Xiangwang and Yang, Tingting and Lan, Dapeng},
  journal={IEEE Trans. Netw. Sci. Eng.},
  title={Blockchain-Aided Secure Access Control for {UAV} Computing Networks}, 
  year={2024},
  month={Nov.--Dec.},
  volume={11},
  number={6},
  pages={5267-5279},
  doi={10.1109/TNSE.2023.3324639}}

@ARTICLE{10949678,
  author={Lin, Xin and Liu, Aijun and others},
  journal={IEEE Trans. Wireless Commun.},
  title={Secure Beamforming and Anti-Jamming Coalition Formation for Air-Terrestrial Integrated Ad-Hoc Networks},
  year={2025},
  volume={24},
  number={8},
  pages={6722-6736},
  doi={10.1109/TWC.2025.3555600}}

@ARTICLE{XJP_MIMOCovert,
	author={Bai, Lin and Xu, Jinpeng and Zhou, Lin},
	journal={IEEE Trans. Commun.}, 
	title={Covert Communication for Spatially Sparse mm{W}ave Massive {MIMO} Channels}, 
	year={Mar. 2023},
	volume={71},
	number={3},
	pages={1615-1630},}

@ARTICLE{11146623,
  author={Xie, Xin and Xu, Jinpeng and Chen, Qun and Gong, Silong and Zhang, Tao},
  journal={IEEE Internet Things J.},
  title={Multiobjective Optimization of Edge Server Placement in {UAV} Ad Hoc Networks},
  year={2026},
  month={Feb.},
  volume={13},
  number={4},
  pages={5682-5693},
  doi={10.1109/JIOT.2025.3605295}
}

@misc{xu2026exploitinguavattitudecovert,
	title={Exploiting {UAV} Attitude for Covert Communications: Dynamics-Consistent Trajectory-Attitude Co-Design}, 
	author={Jinpeng Xu and Lin Zhou and Xin Xie and Mengyuan Zhang and Jingjing Wang},
	url={https://arxiv.org/abs/2610.01307}, 
}

@ARTICLE{Xu2022UAVChannel,
	author={Xu, Jinpeng and Cheng, Xiang and Bai, Lin},
	journal={IEEE Trans. Antennas Propag.}, 
	title={A 3-{D} Space-Time-Frequency Non-Stationary Model for Low-Altitude {UAV} mm{W}ave and Massive {MIMO} Aerial Fading Channels}, 
	year={Nov. 2022},
	volume={70},
	number={11},
	pages={10936-10950},
}

@article{He-CovertUAV-2026,
	author    = {Yingqi He and Jinpeng Xu and Lin Zhou and Jingjing Wang and Chunxiao Jiang},
	title     = {Channel inversion power control aided multi-user secret and covert {UAV} communications},
	journal   = {IEEE Trans. Wireless Commun.},
	volume    = {25},
	pages     = {11002--11018},
	month={Jan.},
	year      = {2026}
}

@ARTICLE{11180147,
  author={Wang, Jingjing and Long, Wei and Liu, Yizhong and Zhang, Xin and Zhang, Zheng and Deng, Robert H.},
  journal={IEEE Trans. Inf. Forensics Security},
  title={A Lightweight Consensus Mechanism for Large-Scale {UAV} Networking},
  year={2025},
  volume={20},
  number={},
  pages={10361-10375},
  doi={10.1109/TIFS.2025.3614501}
}

@INPROCEEDINGS{11591198,
  author={Xu, Jinpeng and Zhang, Shuowen},
  booktitle={Proc. IEEE Global Commun. Conf. Workshops (GLOBECOM Wkshps)},
  title={{MIMO} Radar Meets Polarization-Reconfigurable Antennas: A {BCRB} Perspective},
  year={2025},
  month={Dec.},
  pages={818-823},
  doi={10.1109/GCWkshps68340.2025.11591198}}

@article{en_covert_2025,
title = {Intelligent Covert Communication: Recent Advances and Future Research Trends},
journal = {Engineering},
volume = {44},
pages = {101-111},
year = {2025},
issn = {2095-8099},
doi = {10.1016/j.eng.2024.12.007},
author = {Li, Zan and Shi, Jia and Si, Jiangbo and Lv, Lu and Guan, Lei and Hao, Benjian and Tie, Zhuangzhuang and Wang, Danyang and Xing, Chengwen and Quek, Tony Q. S.}
}

@article{6g_ai_2026,
title = {A Systematic Perspective on Communication Innovations Toward {6G}},
journal = {Engineering},
volume = {31},
pages = {19-24},
year = {2023},
issn = {2095-8099},
doi = {https://doi.org/10.1016/j.eng.2023.03.011},
author = {Ge, Ning and Jiang, Chunxiao and Lu, Jianhua}
}

@article{Xu2026IntelligentFANETs,
	title={Intelligent Foundation Model-Enabled Low-Altitude {FANETs}: Integrating Communication, Control, Sensing, and Security},
	author={Xu, Jinpeng and Zhou, Lin and Zhang, Shuowen and Xie, Xin and Yuan, Weijie and Wang, Jingjing},
	journal={IEEE Internet of Things Mag.},
	volume={9},
	number={2},
	pages={84-92},
	month={Jan.},
	year={2026},
	publisher={IEEE}
}
\end{document}